\documentclass[letterpaper,journal]{IEEEtran}

\usepackage[english]{babel}

\usepackage{graphicx}
\usepackage{xcolor}
\usepackage{amsmath}
\usepackage{bm}
\usepackage{amsfonts}
\usepackage{amssymb}
\usepackage{algorithm}
\usepackage{algpseudocode}
\usepackage{dsfont}
\usepackage{amsthm}
\usepackage{comment}
\usepackage{cite}
\usepackage{booktabs}
\usepackage{placeins}

\newtheorem{theorem}{Theorem}
\newtheorem{lemma}{Lemma}
\newtheorem{corollary}{Corollary}
\newtheorem*{assumption}{Assumption}
\newtheorem*{remark}{Remark}

\title{Age of Information Optimization for Status Updates in Integrated Sensing and Communication Systems}

\author{Marco Zanni, Mohamad Assaad, Touraj Soleymani
\thanks{Marco Zanni (marco.zanni@centralesupelec.fr) and Mohamad Assaad (Mohamad.Assaad@centralesupelec.fr) are with the Laboratory of Signals and Systems, CentraleSup\'{e}lec, University of Paris-Saclay, 91190 Gif-sur-Yvette, France. Touraj Soleymani is with the City St George's School of Science and Technology, University of London, London EC1V~0HB, United Kingdom (touraj.soleymani@citystgeorges.ac.uk).}%
}

\begin{document}
\maketitle

\begin{abstract}
In this paper, we study age of information (AoI) optimization for status updating in an integrated sensing and communication (ISAC) system. We consider a discrete-time architecture in which a base station interacts with a physical environment and a remote monitor, and at each time slot can operate in one of three modes: sensing, communication, or joint sensing and communication. Each mode is unreliable and incurs a different operational cost. The objective is to minimize a discounted infinite-horizon cost that combines the AoI at the monitor with action-dependent sensing and communication costs.
For the single source scenario, we formulate the problem as a Markov decision process with a two-dimensional AoI state and prove that the optimal stationary policy admits an ordered threshold structure in the AoI state space. Since the AoI evolves over an infinite space, we truncate the state space to reduce complexity and rigorously bound the resulting error. The analysis analytically determines the truncation size needed to keep the error below a given threshold. For the multi-source scenario, we formulate the scheduling problem as a restless multi-armed bandit. We develop both a Whittle index policy and an approximate Whittle index policy for scheduling under two different regimes, one where indexability is guaranteed, and one where it is not.
Numerical results illustrate the structure of the optimal policy in the single-source case and show that the proposed approximate Whittle index policy performs comparably to the Whittle index policy in the indexable regime, while remaining effective beyond it.

\end{abstract}


\section{Introduction}\label{sec:introduction}

Real-time remote monitoring systems are a fundamental component of modern cyber-physical infrastructures, allowing remote controllers, operators, and decision-makers to keep awareness of dynamic physical processes. Their importance spans a wide range of applications, including industrial automation, smart transportation, environmental surveillance, and networked control. In such systems, status information is collected from one or more sources and delivered over possibly unreliable communication channels, so that decisions can be made on the basis of the most recent available observations. However, successful delivery alone is not sufficient: the information available at the monitor must also be timely, since outdated status updates may provide an inaccurate representation of the current system state and degrade the quality of subsequent decisions.

A natural metric to quantify this notion of timeliness is the age of information (AoI), introduced in \cite{kaul2012realtime}. If $h(t)$ denotes the generation time of the most recently received update available at a receiver at time $t$, then the AoI is defined as
\[
\Delta(t)=t-h(t).
\]
The AoI measures how much time has passed since the generation of the freshest update currently available at the destination. Since it directly measures information staleness at the receiver, the AoI has become a standard timeliness metric in status updating systems, and it has been extensively studied in queueing systems, scheduling problems, wireless networks, and remote estimation settings \cite{yates2021aoi,sun2017update,kriouile2022whittle}. In these settings, AoI serves as a natural framework for analyzing trade-offs among update frequency, transmission reliability, and resource usage.

The AoI has also stimulated the development of several related metrics that extend the notion of timeliness. The age of incorrect information (AoII) was introduced in \cite{maatouk2023aoii} to account for both staleness and correctness: the AoII grows only when the receiver's estimate of the current system state is incorrect. Likewise, the value of information (VoI) considers how much a new observation improves performance in a control or estimation loop \cite{soleymani2022voi_quant,soleymani2023voi_global}  emphasizes that information should be evaluated not only by whether it is delivered, but also by how useful it is for the task of interest. To this day, the AoI remains the canonical and most straightforward metric when the main objective is to control information freshness.

Among the many problems studied in the AoI literature, scheduling takes a central role. When several users share limited communication resources, the system must decide which user should be served at each time in order to maintain freshness across the network. This question has been investigated in several settings, including broadcast networks, random access systems, and more general AoI optimization problems \cite{kadota2018broadcast_aoi,sun2020closed_form_whittle,tripathi2024whittle_functions,zhou2024active_time_aoi,xu2024aoi_federated_learning,liu2025qaoi_whittle}. In these multi-user settings, the exact dynamic programming solution is often computationally out of reach, which makes low-complexity scheduling rules particularly attractive.

A prominent approach in this direction is given by restless multi-armed bandit formulations and Whittle index policies. Following the seminal work \cite{weber1990restless_bandits}, Whittle index policies have received significant attention in wireless scheduling and, more recently, in the AoI literature. In practice, at each time step a priority index is assigned to every user, and this induces a low-complexity policy where the scheduler can address the users with the highest priority. This heuristic is generally well-performing. In the AoI context, Whittle index policies have been used for broadcast scheduling, random access, federated learning, and query-aware uplink systems \cite{kadota2018broadcast_aoi,sun2020closed_form_whittle,tripathi2024whittle_functions,zhou2024active_time_aoi,xu2024aoi_federated_learning,liu2025qaoi_whittle}. Their performance is not only empirical: in some settings they have also been shown to be asymptotically optimal or even globally optimal, as in \cite{kriouile2022whittle,larranaga2017asymptotically}. The simplicity and strong performance of the Whittle index policy makes it interesting to analyze in our model.

The present paper is motivated by a setting in which freshness is determined by two mechanisms: information about the current state of a physical process must first be acquired, and then delivered to a remote monitor. The monitor cannot directly observe the physical process, but has to rely on a central base station that collects status information on its behalf. This is quite frequent in practice since the monitor cannot have a view of the whole environment required for monitoring. In particular, the base station acquires information about the current state of the physical process through a sensing mechanism, e.g.  by sending a radar signal to collect updated observations about  the environment useful for the physical process.

In such a system, the scheduler at the base station must decide not only when to communicate, but also when to refresh its own local knowledge of the process. The base station can operate in three modes: pure sensing, pure communication, and a joint mode in which it acquires fresh status information and communicates previously acquired information within the same slot. The decision maker must choose among these three competing actions with different costs and success probabilities.  Real-world applications can be, for instance:
\begin{itemize}
    \item \textit{Remote navigation:} a remote operator or controller must track the state of a vehicle and of its surrounding environment, including nearby obstacles. The monitored physical process is the navigation scene, while the base station acquires fresh information through onboard or roadside sensors, and then delivers the acquired status to the remote monitor.

    \item \textit{Industrial robotic cells:} a control room must monitor the state of a production area where mobile robots and human workers coexist. The physical process includes the positions and operating conditions of these elements, while the base station obtains fresh observations from sensors before reporting them to the monitor.
\end{itemize}

An integrated sensing and communication (ISAC) architecture provides a natural operational setting for this problem, since the same platform is used both to acquire and to deliver information \cite{liu2020joint,luong2021rrm_jrc,zhang2025intelligent_isac,wen2025iscc_survey}. Much of the existing ISAC literature focuses on physical-layer metrics such as waveform design, beamforming, interference management, estimation error, and throughput. Our goal is different in that we do not seek to optimize the physical layer operation of an ISAC system, but to characterize the AoI-optimal scheduling rule. AoI and scheduling ISAC layers have received less attention in the literature. Recent examples include AoI optimization in multi-UAV, UAV-enabled, and air-ground ISAC systems \cite{zhou2024multiuav_aoi,bai2025uav_isac_aoi,liu2025joint_aoi_iscc,mei2025aoi_jamming}. \cite{fan2026multimodal} applies AoI metrics to an ISAC setting, but focuses on the physical layer of the system and applies it specifically to vehicular networks. A preliminary two-action formulation of our freshness problem was considered in \cite{soleymani2026freshness_isac}, where the base station can only choose between separate sensing and communication operations. The present paper studies a three-action model by introducing a joint sensing and communication action, which is motivated by the recent settings in ISAC technology in which  a transmitter can transmit a communication signal that can be useful to collect sensing information. 
Unlike the two-mode scenario, the third joint communication-and-sensing mode fundamentally changes the problem geometry, making standard submodularity approaches inapplicable. Furthermore, this paper extends the analysis to multiple physical process–monitor pairs. 

\begin{figure}
    \centering
    \includegraphics[width=0.5\linewidth]{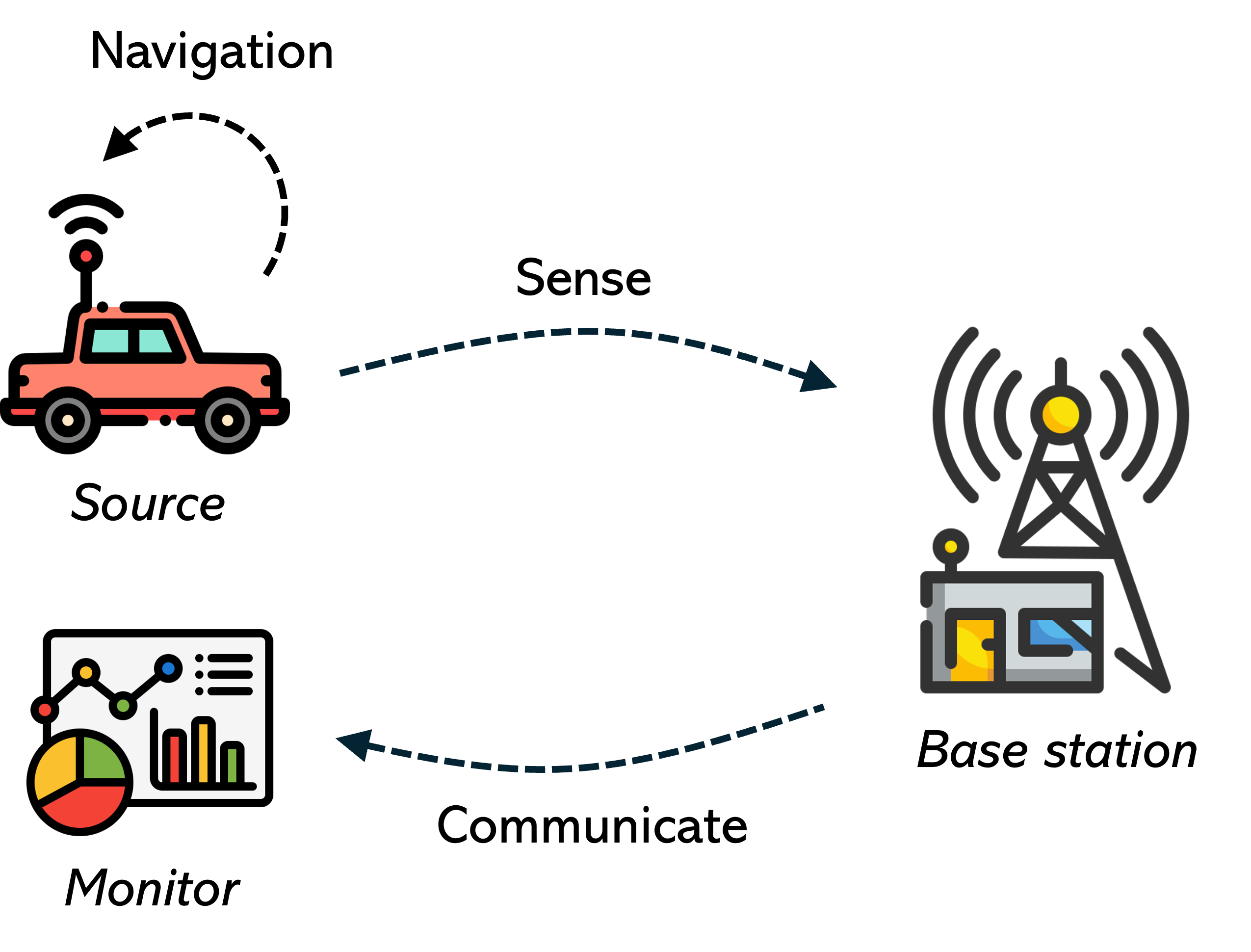}
    \caption{Representative example of an ISAC architecture for remotely monitoring a ground vehicle.}
    \label{fig:ISAC-model}
\end{figure}

In this paper, we study AoI optimization for status updates in such an architecture. We consider a discrete-time system composed of a source, an ISAC-enabled base station, and a remote monitor (see Figure \ref{fig:ISAC-model} for a representative example). The source can track the state of an underlying physical process, while the monitor cannot, and must rely on status information collected by the base station. At every time step, the base station can perform sensing to acquire fresh information from the source about the current state of the process, communicate previously acquired status information to the monitor, or perform sensing and communication simultaneously. These operations are unreliable and incur different costs. To capture the resulting trade-off, we formulate a discounted infinite-horizon Markov decision process whose state is the pair formed by the AoI at the monitor and the AoI at the base station. The objective is to minimize a long-term cost that combines information staleness at the monitor with action costs. We also extend the analysis to a multiple process-monitor setting in which one base station must schedule several monitored processes.

\subsection{Contributions and Organization}

The main contribution of this paper is a structural and algorithmic study of the above problem in both single process-monitor and multiple process-monitor pairs scenarios.
\begin{itemize}
    \item For the single process-monitor case, we show that the optimal stationary policy admits an ordered switching structure in the AoI state space. This yields a policy described by two switching thresholds. We also quantify the error induced by truncating the unbounded state space and derive a simple closed-form criterion for selecting the truncation level for the numerical computation of the optimal policy. The proof of the threshold structure of the optimal policy significantly differs from standard AoI and/or MDP formulations, since the optimal value function is not submodular in our problem.
    \item We then extend the analysis to a multiple process-monitor pairs scenario in which one base station must share its sensing and communication capability among several subsystems. This leads to a restless multi-armed bandit formulation. By introducing an idle action and applying a Lagrangian relaxation, we obtain a relaxed single-arm problem that enables the construction of Whittle-type scheduling rules. We provide a sufficient condition under which the relaxed problem is Whittle-indexable, so that an exact Whittle index policy can be defined. In contrast to classical AoI scheduling formulations, each arm involves two coupled AoI variables and three active sensing-communication modes, so the scheduling rule must jointly determine which sources to activate and which ISAC action to choose.
    \item Finally, we develop an approximate Whittle index policy based on linear interpolation. This approximate construction is computationally attractive, can be used in non-indexable regimes, and comes with an explicit approximation bound. Numerical results illustrate both the threshold geometry of the optimal single process-monitor policy and the effectiveness of the proposed heuristic policies in the multiple process-monitor pairs scenario.
\end{itemize}

The remainder of the paper is organized as follows. Section \ref{sec:prolem-formulation} introduces the system model and formulates the optimization problem. Section \ref{sec:single-source-scenario} studies the single-source scenario and establishes the structure of the optimal policy together with the truncation bound. Section \ref{sec:multi-source-scenario} addresses the multi-source scenario and develops the index-based scheduling policies. Section \ref{sec:numerical-results} presents numerical results. Section \ref{sec:conclusions} concludes the paper.

\section{Problem Formulation}\label{sec:prolem-formulation}

We study a discrete-time remote monitoring system supported by an integrated sensing and communication (ISAC) infrastructure. The architecture consists of three entities: a physical process, an ISAC-enabled base station, and a remote monitor. The monitor aims to track the evolution of the process state, useful for the monitor that has a  limited sensing capability and does not have a good view/observation of the environment. The monitor must rely on information provided by the base station. At each time slot, the base station selects one of three operating modes: sensing the current process state, transmitting previously acquired status information to the monitor, or performing sensing and communication simultaneously (i.e; acquiring new sensing while transmitting the previous sensing status to the monitor). The transmission channels are lossy, so the chosen operation may fail. The objective is to design a scheduling policy for the base station that balances information freshness and operational cost, namely by keeping the AoI at the monitor low while accounting for sensing and communication costs.

Our focus is on the decision and scheduling layer of an ISAC monitoring system, rather than on physical layer design. Accordingly, sensing, communication, and joint sensing and communication are modeled through success probabilities and operational costs. In particular, the joint action is treated as a single effective mode: when it succeeds, it both delivers the previously sensed information and refreshes the base station information; when it fails, neither update is performed. This abstraction allows us to isolate the effect of the additional reset mechanism on the AoI dynamics.

\subsection{System Model}

Let $Z_k$ denote the state of the physical process/source at time $k$. At the beginning of each slot, the base station selects an action
\[
u_k\in\{\mathsf{sense},\mathsf{comm},\mathsf{joint}\},
\]
corresponding, respectively, to sensing the process/source, transmitting previously acquired information to the monitor, or performing both operations simultaneously. These three actions incur fixed costs $c_0\geq 0$, $c_1\geq 0$, and $c_2\geq 0$.

The key difference among the three modes lies in how information is handled within a slot. A sensing action attempts to acquire the current process/source state. A communication action attempts to forward the most recent state information already available at the base station. A joint sensing and communication action combines these two operations: during the same slot, the base station transmits its previously available estimate while also attempting to collect a fresh measurement of the source.

Let $X_k$ denote the measurement obtained by the base station at time $k$, whenever sensing is performed successfully.

If $u_k=\mathsf{sense}$, the base station sends a radar signal in order to collect fresh information regarding the state of the process. The sensing operation succeeds with probability $\lambda_0\in(0,1)$, and the sensing outcome satisfies
\[
\Pr(X_k = Z_k \mid Z_k, u_k = \mathsf{sense}) = \lambda_0,
\]
\[
\Pr(X_k = \emptyset \mid Z_k, u_k = \mathsf{sense}) = 1 - \lambda_0.
\]
The base station stores the most recent successfully sensed state. Let $\tilde Z_k$ denote the state information stored at the base station after the ISAC action at time $k$.

If $u_k=\mathsf{comm}$, the base station attempts to transmit the most recent locally available estimate, denoted by $\tilde Z_{k-1}$, to the remote monitor. The communication operation succeeds with probability $\lambda_1\in(0,1)$. Let $Y_k$ denote the received packet at the remote monitor. Then
\[
\Pr(Y_k = \tilde Z_{k-1} \mid \tilde Z_{k-1}, u_k = \mathsf{comm}) = \lambda_1,
\]
\[
\Pr(Y_k = \emptyset \mid \tilde Z_{k-1}, u_k = \mathsf{comm}) = 1 - \lambda_1.
\]

If $u_k=\mathsf{joint}$, the base station transmits the estimate available from the previous slot while simultaneously senses the current source state. In this case, the joint operation succeeds with probability $\lambda_2\in(0,1)$, and the outcome satisfies
\[
\Pr(X_k = Z_k, Y_k = \tilde Z_{k-1} \mid Z_k, \tilde Z_{k-1}, u_k = \mathsf{joint}) = \lambda_2,
\]
\[
\Pr(X_k = \emptyset, Y_k = \emptyset \mid Z_k, \tilde Z_{k-1}, u_k = \mathsf{joint}) = 1 - \lambda_2.
\]
The state information stored at the base station evolves according to
\[
\tilde Z_k =
\begin{cases}
X_k, & \text{if } u_k \in \{\mathsf{sense},\mathsf{joint}\} \text{ and } X_k \neq \emptyset,\\
\tilde Z_{k-1}, & \text{otherwise.}
\end{cases}
\]

Let $\hat Z_k$ denote the state estimate available at the remote monitor after the ISAC action at time $k$. The monitor can be updated only when an action involving communication is selected and the transmission succeeds. Its state therefore evolves according to
\[
\hat Z_k =
\begin{cases}
Y_k, & \text{if } u_k \in \{\mathsf{comm}, \mathsf{joint}\} \text{ and } Y_k \neq \emptyset,\\
\hat Z_{k-1}, & \text{otherwise.}
\end{cases}
\]
In particular, under a pure sensing action the monitor receives no packet and keeps its previous estimate.

For notational convenience, we also introduce the binary outcome variable $\eta_k\in\{\mathsf{succ},\mathsf{fail}\}$, which indicates whether the selected operation at time $k$ is successful. Its conditional distribution is given by
\begin{equation}
\Pr(\eta_k = \mathsf{succ} \mid u_k) =
\begin{cases}
\lambda_0, & \text{if } u_k = \mathsf{sense},\\
\lambda_1, & \text{if } u_k = \mathsf{comm},\\
\lambda_2, & \text{if } u_k = \mathsf{joint},
\end{cases}
\end{equation}
with $\Pr(\eta_k = \mathsf{fail} \mid u_k) = 1 - \Pr(\eta_k = \mathsf{succ} \mid u_k)$.

This model separates three distinct uses of the ISAC resource: information acquisition through sensing, information delivery through communication, and the combined execution of the two within the same slot. Throughout the paper, we assume that $\lambda_2\leq \lambda_0\leq \lambda_1$ and $c_0\leq c_1\leq c_2$. These assumptions reflect the fact that communication is typically more reliable than sensing because of coding and retransmission mechanisms, whereas simultaneous sensing and communication is generally the most demanding operating mode in terms of both reliability and cost.

\subsection{Freshness Metric}

We track information freshness separately at the remote monitor and at the base station. For $i \in \{m,b\}$, where $m$ and $b$ index the remote monitor and the base station, let $\alpha^i_k$ denote the AoI at time $k$ before the ISAC action, and $\alpha^i_{k^+}$ denote the AoI at time $k$ after the ISAC action. These variables quantify the freshness of the state information $\hat Z_k$ available at the remote monitor and $\tilde Z_k$ stored at the base station.

The post-action AoI values depend on which operation is selected and on whether that operation succeeds. In particular, when the selected action is successful, the pair $(\mathrm{AoI}^m_{k+}, \mathrm{AoI}^b_{k+})$ evolves as follows:
\begin{equation}
(\mathrm{AoI}^m_{k+}, \mathrm{AoI}^b_{k+})=
\begin{cases}
(\mathrm{AoI}^m_k, 0), & \text{if } u_k=\mathsf{sense},\\
(\mathrm{AoI}^b_k, \mathrm{AoI}^b_k), & \text{if } u_k=\mathsf{comm},\\
(\mathrm{AoI}^b_k, 0), & \text{if } u_k=\mathsf{joint}.
\end{cases}
\end{equation}

If the selected operation is unsuccessful, no fresh information is acquired or delivered, and therefore
\[
(\mathrm{AoI}^m_{k+}, \mathrm{AoI}^b_{k+})=
(\mathrm{AoI}^m_k, \mathrm{AoI}^b_k).
\]

Finally, between time instants $k+$ and $k+1$, the AoI increases by one unit at both entities, so that
\[
\mathrm{AoI}^i_{k+1} = \mathrm{AoI}^i_{k+} + 1.
\]

This representation highlights the fact that the remote monitor and the base station may carry information with different freshness levels, depending on whether the selected action refreshes local information, delivered information, or both.

\subsection{Performance Criterion and Optimization Problem}

The dynamics introduced above induce a discounted infinite-horizon Markov decision process. At each time $k$, the ISAC system selects an action $u_k$. For compactness, we use the conventions $\alpha^i_k:=\mathrm{AoI}^i_k$, $u_k=0 \Leftrightarrow u_k=\mathsf{sense}$, $u_k=1 \Leftrightarrow u_k=\mathsf{comm}$, $u_k=2 \Leftrightarrow u_k=\mathsf{joint}$, $\eta_k=0 \Leftrightarrow \eta_k=\mathsf{fail}$, and $\eta_k=1 \Leftrightarrow \eta_k=\mathsf{succ}$. The state at time $k$ is then
\[
S_k=(\alpha^m_k,\alpha^b_k).
\]

We associate with each state-action pair a one-step cost composed of a freshness term at the remote monitor and an action-dependent operational cost. Specifically, the stage cost is defined as
\begin{equation}\label{eq:stage-cost-base}
\begin{aligned}
g(S_k,u_k)&=\alpha^m_k+c_0\mathds{1}\{u_k=0\}\\
&\quad+c_1\mathds{1}\{u_k=1\}+c_2\mathds{1}\{u_k=2\}.
\end{aligned}
\end{equation}

Starting from an initial state $S_0$, the objective is to find a policy $\pi$ that minimizes the expected discounted cumulative cost over an infinite horizon, namely
\begin{equation}
	\min_{\pi \in \mathcal P} \mathbb{E}\left[\sum_{k=0}^{\infty} \gamma^k g(S_k,u_k)\right],
	\label{eq:problem-formulation}
\end{equation}
subject to the state dynamics defined in the previous subsections, where $\gamma$ is a discount factor, and $\mathcal P$ is the set of admissible stationary policies. We denote by $\pi^*$ an optimal policy.

The problem in \eqref{eq:problem-formulation} captures the trade-off of the considered ISAC system. On the one hand, frequent updates improve freshness at the remote monitor by reducing the age of the information on which it relies. On the other hand, each operating mode consumes resources and is affected by a different reliability level. The optimization problem seeks a policy that coordinates sensing and communication decisions while balancing freshness performance against sensing and communication costs.

For the structural analysis developed in the next section, we restrict the discount factor $\gamma$ to the following admissible regime.

\begin{assumption}\label{ass:admissible-gamma}
The discount factor $\gamma$ is chosen such that
\begin{equation}\label{eq:parameters-assumption}
0 < \gamma \leq \frac{\lambda_2}{\lambda_0\lambda_1+\lambda_2(1-\lambda_1)}.
\end{equation}
\end{assumption}

\begin{remark}
This condition is used only to establish the monotonicity of the action-difference function $\Delta_{02}^\ast$ in Lemma \ref{lemma:delta-non-decreasing-in-a_s}. It is not a physical constraint on the sensing or communication links, and it is not required for the MDP formulation, for the existence of an optimal stationary policy, or for the numerical computation of the optimal policy. Numerically, the ordered threshold structure is also observed in several sets of parameters violating this sufficient
condition.

\end{remark}

\section{Single process-monitor Scenario}\label{sec:single-source-scenario}

In this section, we specialize the model in \eqref{eq:problem-formulation} to the single process-monitor case. We study the structure of the optimal stationary policy, with the goal of showing that it has a switching-threshold form in the AoI state space. We then find an upper bound for the error caused by the truncation of the state space, which is inevitable for the offline computation of the optimal policy.

\subsection{State Transitions and Reachable State Space}
We first characterize the state transitions and the reachable area of the AoI state space. With the conventions introduced in Section \ref{sec:prolem-formulation}, the single process-monitor system is a discounted infinite-horizon MDP with state $S_k=(\alpha^m_k,\alpha^b_k)$ and action space $U=\{0,1,2\}$. The state dynamics follow directly from the AoI update rules. In particular,
\begin{equation}\label{eq:state-evolution}
(\alpha^m_{k+1},\alpha^b_{k+1})=
\begin{cases}
	(\alpha^m_k+1,1), & \text{if } (u_k,\eta_k)=(0,1),\\
	(\alpha^m_k+1,\alpha^b_k+1), & \text{if } (u_k,\eta_k)=(0,0),\\
	(\alpha^b_k+1,\alpha^b_k+1), & \text{if } (u_k,\eta_k)=(1,1),\\
	(\alpha^m_k+1,\alpha^b_k+1), & \text{if } (u_k,\eta_k)=(1,0),\\
	(\alpha^b_k+1,1), & \text{if } (u_k,\eta_k)=(2,1),\\
	(\alpha^m_k+1,\alpha^b_k+1), & \text{if } (u_k,\eta_k)=(2,0).
\end{cases}
\end{equation}
We make the process start at the initial state $S_0=(1,1)$\footnote{One could also fix the initial state at (0,0), in which case the state transitions would all lead to (1,1) in the following time step, for whatever action $u$ and outcome $\eta$. The results would be identical, up to a shift in the indices.}. From this, only a triangular subset of $\mathbb N^2$ is reachable: every transition preserves $\alpha^m_k\geq\alpha^b_k$ for any time step $k$. Accordingly, the analysis can be restricted to the reachable state space
\[
\mathcal S:=\{(\alpha^m,\alpha^b)\in\mathbb N^2:\alpha^m\ge \alpha^b\ge 1\}.
\]

\begin{theorem}\label{th:policy-structure}
The problem in \eqref{eq:problem-formulation} admits an optimal stationary deterministic policy $\pi^*$ with a switching-threshold structure. Specifically, there exist two functions $\tau_1, \tau_2 : \mathbb{N}\rightarrow\mathbb{N}_0\cup\{+\infty\}$ such that
\[
u^*(\alpha^m,\alpha^b)=
\begin{cases}
	\mathsf{sense,} & \text{if } \alpha^m\leq\tau_1(\alpha^b),\\
	\mathsf{joint,} & \text{if } \tau_1(\alpha^b)<\alpha^m\leq\tau_2(\alpha^b),\\
	\mathsf{comm,} & \text{if } \alpha^m>\tau_2(\alpha^b).\\
\end{cases}
\]
for all $(\alpha^m,\alpha^b)\in\mathcal{S}$, where $\tau_1(\alpha^b)$ is nondecreasing in $\alpha^b$ and $\tau_1(\alpha^b)\leq\tau_2(\alpha^b)$ for every $\alpha^b$.
\end{theorem}

Theorem 1 shows that the optimal policy has an ordered structure: along each fixed value of the base station AoI, the optimal action can only move in the order $0 \to 2 \to 1$. Thus, the optimal rule is described by two switching boundaries.

The derivation proceeds in several steps. We first present the space in which the Bellman operator can be defined. Then, we derive structural properties of the Bellman operator and of the optimal value function. These properties are subsequently used to control directional increments of the optimal value function and to analyze the action-difference functions associated with the three available actions. The goal is to prove that the action-difference functions satisfy suitable single-crossing properties. Finally, we combine these findings to prove that the optimal policy admits a threshold structure.

\begin{remark}
Unlike several standard AoI problems, the present model does not naturally admit a proof based on submodularity of the optimal value function, which makes the derivation more delicate. The argument developed below relies instead on monotonicity and concavity properties of the Bellman operator and suitable bounds on directional increments.
\end{remark}

\subsection{Proof of Theorem \ref{th:policy-structure}}

We begin by introducing a setting in which the Bellman operator can be analyzed in a rigorous way. Since the MDP is unbounded, it is necessary to define a suitable norm and work on the corresponding Banach space. This allows us to show that the Bellman operator is well defined and contractive, and therefore admits a unique fixed point, which coincides with the optimal value function. Convergence to a fixed point first ensures that the optimal value function exists, and second lets us prove some of its key properties through value iteration.

Let $\rho\in(1,\gamma^{-1})$, and define $w(\alpha^m,\alpha^b):=\rho^{\alpha^m}$, the norm
\[
\|V\|_w:=\sup_{(\alpha^m,\alpha^b)\in\mathcal S}\frac{|V(\alpha^m,\alpha^b)|}{w(\alpha^m,\alpha^b)},
\]
and the Banach space
\[
\mathcal B_w:=\{V:\mathcal S\to\mathbb R:\|V\|_w<\infty\}.
\]
For $V\in\mathcal B_w$, define the Bellman operator
\[
\begin{aligned}
TV(S)&=\min_{u\in\{0,1,2\}} Q_u(S),\\
Q_u(S)&:=g(S,u)+\gamma\mathbb E\left[V(S')\mid S,u\right].
\end{aligned}
\]
Since the next monitor AoI always satisfies $\alpha^{m\prime}\le \alpha^m+1$, we have
\[
\mathbb E\bigl[w(S')\mid S,u\bigr]\le \rho\, w(S),\qquad \forall S\in\mathcal S,\ \forall u\in\{0,1,2\}.
\]
Moreover, $g(S,u)\le \alpha^m+c_2\le C_w w(S)$ for some finite constant $C_w>0$, because $\sup_{n\ge 1}(n+c_2)\rho^{-n}<\infty$. Therefore $T:\mathcal B_w\to\mathcal B_w$ and, for every $V,W\in\mathcal B_w$,
\[
\|TV-TW\|_w\le \gamma\rho\,\|V-W\|_w.
\]
Since $\gamma\rho<1$, $T$ is a contraction on $\mathcal B_w$. We work on the Banach space $\mathcal B_w$, on which the Bellman operator is a contracting operator. This lets us exploit a value iteration algorithm for which $T$ has a unique fixed point, which is the optimal value function $V^*$. We can now write the action-value function
\begin{equation}
    Q^*_u(S)=g_u(S,u)+\gamma\mathbb{E}\bigl[V^*(S')|S,u\bigr]
\end{equation}
and the Bellman equation
\begin{equation}
    V^*(S)=\min_{u\in\{0,1,2\}}Q^*_u(S).
\end{equation}

We now prove some structural properties of the optimal value function by defining a class of functions that is invariant under the Bellman operator. These properties will be used later on to analyze the behavior of the action-value functions $Q_u^*$.

Let $\mathcal F$ be the class of functions $V:\mathcal S\to\mathbb R$ with these properties:
\begin{itemize}
    \item $V$ is coordinatewise nondecreasing;
    \item $V(\alpha^m, \alpha^b)$ is discretely concave in $\alpha^m$ for every fixed $\alpha^b$.
\end{itemize}

\begin{lemma}\label{lemma:TV-nondecreasing}
Let $V\in \mathcal F$ be coordinatewise nondecreasing. Then $TV$ is also coordinatewise nondecreasing.
\end{lemma}

\begin{proof}
For $u\in\{0,1,2\}$, write
\[
Q_u(s)=g(s,u)+\gamma\,\mathbb E[V(S')\mid s,u].
\]
Using the state transition in \eqref{eq:state-evolution}, we obtain
\begin{align*}
Q_0(\alpha^m,\alpha^b)
&=\alpha^m+c_0+\gamma\bigl[\lambda_0 V(\alpha^m+1,1)\\
&\qquad +(1-\lambda_0)V(\alpha^m+1,\alpha^b+1)\bigr],\\
Q_1(\alpha^m,\alpha^b)
&=\alpha^m+c_1+\gamma\bigl[\lambda_1 V(\alpha^b+1,\alpha^b+1)\\
&\qquad +(1-\lambda_1)V(\alpha^m+1,\alpha^b+1)\bigr],\\
Q_2(\alpha^m,\alpha^b)
&=\alpha^m+c_2+\gamma\bigl[\lambda_2 V(\alpha^b+1,1)\\
&\qquad +(1-\lambda_2)V(\alpha^m+1,\alpha^b+1)\bigr].
\end{align*}

Because $V$ is coordinatewise nondecreasing, the sum in the square bracket is coordinatewise nondecreasing for every $u$. Moreover, the stage cost $g\bigl((\alpha^m,\alpha^b),u\bigr)=\alpha^m+c_u$ is nondecreasing in the state. The pointwise minimum of coordinatewise nondecreasing functions is coordinatewise nondecreasing. Therefore, $TV(s)=\min_{u\in\{0,1,2\}}Q_u(s)$ is coordinatewise nondecreasing.
\end{proof}

\begin{lemma}\label{lemma:rowwise-discrete-concavity}
Let $V\in \mathcal F$ be discretely concave in $\alpha^m$. Then $TV$ is also discretely concave in $\alpha^m$.
\end{lemma}

\begin{proof}
Fix $\alpha^b \ge 1$ and, for $\alpha^m \ge \alpha^b$, define
\[
q_u(\alpha^m) := Q_u(\alpha^m,\alpha^b), \qquad u \in \{0,1,2\}.
\]
Using the state transitions in \eqref{eq:state-evolution}, we write their forward differences as
\begin{align*}
q_0(\alpha^m+1)&-q_0(\alpha^m)\\
&=1+\gamma\lambda_0\bigl[V(\alpha^m+2,1)-V(\alpha^m+1,1)\bigr]\\
&\quad+\gamma(1-\lambda_0)\Bigl[V(\alpha^m+2,\alpha^b+1)\\
&\qquad -V(\alpha^m+1,\alpha^b+1)\Bigr],\\
q_1(\alpha^m+1)&-q_1(\alpha^m)\\
&=1+\gamma(1-\lambda_1)\Bigl[V(\alpha^m+2,\alpha^b+1)\\
&\qquad -V(\alpha^m+1,\alpha^b+1)\Bigr],\\
q_2(\alpha^m+1)&-q_2(\alpha^m)\\
&=1+\gamma(1-\lambda_2)\Bigl[V(\alpha^m+2,\alpha^b+1)\\
&\qquad -V(\alpha^m+1,\alpha^b+1)\Bigr].
\end{align*}
Since $V$ is discretely concave in $\alpha^m$, for every fixed $\alpha^b$ the differences in the square brackets are all nonincreasing in $\alpha^m$. Therefore, for each $u \in \{0,1,2\}$, the sequence $q_u(\alpha^m+1)-q_u(\alpha^m)$ is nonincreasing in $\alpha^m$, which means that $q_u$ is discretely concave. The pointwise minimum of discretely concave functions is discretely concave\footnote{Here discrete concavity means $f(a+2)-2f(a+1)+f(a)\le 0$. Unlike the continuous case, the pointwise minimum preserves this property on $\mathbb N$: if $h(a)=\min_i f_i(a)$ and $i^\star$ reaches the minimum at $a+1$, then $2h(a+1)=2f_{i^\star}(a+1)\ge f_{i^\star}(a)+f_{i^\star}(a+2)\ge h(a)+h(a+2)$.}. Therefore, $TV(\alpha^m,\alpha^b) = \min_{u\in\{0,1,2\}} q_u(\alpha^m)$ is discretely concave in $\alpha^m$ for every fixed $\alpha^b$.
\end{proof}

\begin{lemma}\label{lemma:Vstar-in-F}
The optimal value function $V^*$ belongs to $\mathcal F$.
\end{lemma}

\begin{proof}
Start value iteration from $V^{(0)}\equiv 0\in\mathcal F$. Lemma \ref{lemma:TV-nondecreasing} and Lemma \ref{lemma:rowwise-discrete-concavity} imply $V^{(n+1)}=TV^{(n)}\in\mathcal F$ whenever $V^{(n)}\in\mathcal F$. Hence $V^{(n)}\in\mathcal F$ for all $n\ge 0$. Pointwise convergence of $V^{(n)}$ to $V^*$ ensures $V^*\in\mathcal F$.
\end{proof}

\subsubsection{Marginal Value Increments}
The proof of the threshold structure relies on comparing the three action-value functions $Q_0^*$, $Q_1^*$, and $Q_2^*$. These comparisons involve differences of the optimal value function evaluated at adjacent states. For this reason, we begin by deriving bounds on several one-step increments of $V^*$. These bounds will later be used to show that the action-difference functions satisfy single-crossing properties.

We define the following horizontal, diagonal, and vertical one-step increments:
\begin{align}
A(\alpha^b) 
&:= V^*(\alpha^b+2,1)-V^*(\alpha^b+1,1),\\
B(\alpha^b) 
&:= V^*(\alpha^b+1,\alpha^b+1)-V^*(\alpha^b,\alpha^b),\\
C(\alpha^m,\alpha^b) 
&:= V^*(\alpha^m+1,\alpha^b+2)-V^*(\alpha^m+1,\alpha^b+1),\notag\\
&\hspace{2em}\alpha^m\geq \alpha^b+1.
\end{align}

Lemmas \ref{lemma:Vstar-horizontal-lower-bound}--\ref{lemma:ordered-cost-vertical-bound} are structured as follows. Lemma \ref{lemma:Vstar-horizontal-lower-bound} gives a lower bound on horizontal increments of $V^*$. Lemma \ref{lemma:Vstar-local-upper-bound} gives a uniform upper bound on local one-step increments. Lemma \ref{lemma:ordered-cost-diagonal-zero} identifies the optimal action on the diagonal and yields a recursion for the diagonal increment $B$. Finally, Lemmas \ref{lemma:ordered-cost-diagonal-gap} and \ref{lemma:ordered-cost-vertical-bound} compare the increments $A$, $B$, and $C$, which will be needed to establish the monotonicity of the action-difference functions. The proofs of these lemmas are in the Appendices \ref{app:Vstar-horizontal-lower-bound}--\ref{app:ordered-cost-vertical-bound}.

\begin{lemma}\label{lemma:Vstar-horizontal-lower-bound}
For every $\alpha^m\geq \alpha^b\ge 1$,
\begin{equation}\label{eq:Vstar-horizontal-lower-bound}
V^*(\alpha^m+1,\alpha^b)-V^*(\alpha^m,\alpha^b)\geq \frac{1}{1-\gamma+\gamma\lambda_1}.
\end{equation}
\end{lemma}

\emph{Proof.} See Appendix \ref{app:Vstar-horizontal-lower-bound}.

\begin{lemma}\label{lemma:Vstar-local-upper-bound}
For every $\alpha^m\geq \alpha^b\geq 1$ and every $(\tilde\alpha^m,\tilde\alpha^b)\in\mathcal S$ such that $0\leq \tilde\alpha^m-\alpha^m\leq 1$ and $0\leq \tilde\alpha^b-\alpha^b\leq 1$,
\begin{equation}\label{eq:Vstar-local-upper-bound}
0\leq V^*(\tilde\alpha^m,\tilde\alpha^b)-V^*(\alpha^m,\alpha^b)\leq \frac{1}{1-\gamma}.
\end{equation}
\end{lemma}

\emph{Proof.} See Appendix \ref{app:Vstar-local-upper-bound}.

\begin{lemma}\label{lemma:ordered-cost-diagonal-zero}
For every $\alpha\geq 1$, action $0$ is optimal on the diagonal state $(\alpha,\alpha)$.
\end{lemma}

\emph{Proof.} See Appendix \ref{app:ordered-cost-diagonal-zero}.

Writing the Bellman equation at $(\alpha^b+1,\alpha^b+1)$ and $(\alpha^b,\alpha^b)$, and subtracting the latter from the former, we obtain the following recursive definition for $B(\alpha^b)$ which will be used in Lemmas \ref{lemma:ordered-cost-diagonal-gap} and \ref{lemma:ordered-cost-vertical-bound}: 
\begin{equation}\label{eq:ordered-cost-diagonal-recursion}
B(\alpha^b)=1+\gamma\lambda_0A(\alpha^b)+\gamma(1-\lambda_0)B(\alpha^b+1).
\end{equation}

\begin{lemma}\label{lemma:ordered-cost-diagonal-gap}
For every $\alpha^b\geq 1$,
\begin{equation}\label{eq:ordered-cost-gap-bound}
B(\alpha^b)-A(\alpha^b)\leq
\frac{1-(1-\gamma)A(\alpha^b)}{1-\gamma+\gamma\lambda_0} < \frac{1}{\gamma\lambda_0}.
\end{equation}
\end{lemma}

\emph{Proof.} See Appendix \ref{app:ordered-cost-diagonal-gap}.

\begin{lemma}\label{lemma:ordered-cost-vertical-bound}
For every $\alpha^b \ge 1$ and every $\alpha^m \ge \alpha^b+1$,
\[
C(\alpha^m,\alpha^b)\le B(\alpha^b+1).
\]
\end{lemma}

\emph{Proof.} See Appendix \ref{app:ordered-cost-vertical-bound}.

\subsubsection{Properties of the Action-Difference Functions}
We now define the action-difference functions associated with the three available actions. Their explicit representation will be the main tool for comparing actions across the AoI state space.
\[
\Delta_{01}^*:=Q_0^*-Q_1^*,\qquad
\Delta_{02}^*:=Q_0^*-Q_2^*,\qquad
\Delta_{21}^*:=Q_2^*-Q_1^*.
\]
With this convention, $\Delta^*_{01}\leq 0$ means that sensing is no worse than communication, $\Delta^*_{02}\leq 0$ means that sensing is no worse than the joint action, and $\Delta^*_{21}\leq 0$ means that the joint action is no worse than communication. The idea is to show that these functions are single-crossing in the two coordinates: this leads to the single-switching structure of the optimal policy.

Using \eqref{eq:state-evolution}, the action-difference functions can be written as
\begin{subequations}\label{eq:deltas}
\begin{equation}
\begin{aligned}
\Delta^*_{01}(\alpha^m,\alpha^b)=&\;(c_0-c_1)+\gamma\lambda_0V^*(\alpha^m+1,1)\\
&-\gamma\lambda_1V^*(\alpha^b+1,\alpha^b+1)\\
&+\gamma(\lambda_1-\lambda_0)V^*(\alpha^m+1,\alpha^b+1),
\end{aligned}
\end{equation}

\begin{equation}
\begin{aligned}
\Delta^*_{02}(\alpha^m,\alpha^b)=&\;(c_0-c_2)+\gamma\lambda_0V^*(\alpha^m+1,1)\\
&-\gamma\lambda_2V^*(\alpha^b+1,1)\\
&+\gamma(\lambda_2-\lambda_0)V^*(\alpha^m+1,\alpha^b+1),
\end{aligned}
\end{equation}

\begin{equation}
\begin{aligned}
\Delta^*_{21}(\alpha^m,\alpha^b)=&\;(c_2-c_1)+\gamma\lambda_2V^*(\alpha^b+1,1)\\
&-\gamma\lambda_1V^*(\alpha^b+1,\alpha^b+1)\\
&+\gamma(\lambda_1-\lambda_2)V^*(\alpha^m+1,\alpha^b+1).
\end{aligned}
\end{equation}
\end{subequations}

The next step is to show that the action comparisons vary monotonically over the AoI state space. For fixed $\alpha^b$, we prove that the action-difference functions are nondecreasing in $\alpha^m$. Hence, once sensing becomes worse than another action as the monitor AoI grows, it cannot become better again. This is the single-crossing property that leads to thresholds along each horizontal row of the state space.

\begin{lemma}\label{lemma:delta-non-decreasing-in-a_s}
For each fixed $\alpha^b$, the functions $\Delta^*_{01}(\alpha^m,\alpha^b)$, $\Delta^*_{02}(\alpha^m,\alpha^b)$, and $\Delta^*_{21}(\alpha^m,\alpha^b)$ are nondecreasing in $\alpha^m$.
\end{lemma}

\begin{proof}
See Appendix \ref{app:delta-non-decreasing-in-a_s}.
\end{proof}

We also need to understand how the sensing region changes when the base station AoI increases. The following lemma shows that the differences comparing sensing with the other two actions are nonincreasing in $\alpha^b$. Therefore, a larger base station AoI makes sensing relatively more attractive, which will imply that the lower threshold $\tau_1(\alpha^b)$ is nondecreasing.

\begin{lemma}\label{lemma:delta-non-increasing-in-a_b}
For each fixed $\alpha^m$, the functions $\Delta^*_{01}(\alpha^m,\alpha^b)$ and $\Delta^*_{02}(\alpha^m,\alpha^b)$ are nonincreasing in $\alpha^b$ for $1\leq \alpha^b\leq \alpha^m-1$.
\end{lemma}

\begin{proof}
See Appendix \ref{app:delta-non-increasing-in-a_b}.
\end{proof}

The logical structure is now the following. Lemmas \ref{lemma:Vstar-horizontal-lower-bound}--\ref{lemma:ordered-cost-vertical-bound} provide the increment bounds needed to prove the single-crossing properties of Lemmas \ref{lemma:delta-non-decreasing-in-a_s} and \ref{lemma:delta-non-increasing-in-a_b}. Lemma \ref{lemma:delta-non-decreasing-in-a_s} implies that, for each fixed $\alpha^b$, the sensing region is an initial segment and the communication region is a terminal segment in $\alpha^m$. The remaining states therefore form the intermediate joint region. Lemma \ref{lemma:delta-non-increasing-in-a_b} then implies that the lower sensing threshold is nondecreasing in $\alpha^b$.

\subsubsection{Threshold Structure of the Optimal Policy}
We collect the properties derived from the previous lemmas to prove the optimal policy structure established in Theorem \ref{th:policy-structure}. The single-crossing properties stated above imply that, for every fixed value of the base station AoI, the regions in which the three actions are optimal must appear in an ordered way. This yields the desired switching-threshold structure and allows us to show that the lower switching boundary is nondecreasing.

\begin{proof}[Proof of Theorem \ref{th:policy-structure}]
For each state $(\alpha^m,\alpha^b)$, define
\[
u^*(\alpha^m,\alpha^b)=
\begin{cases}
0, & \begin{array}[t]{l}
\text{if } \Delta^*_{01}(\alpha^m,\alpha^b)\le 0\text{ and}\\
\Delta^*_{02}(\alpha^m,\alpha^b)\le 0,
\end{array}\\[0.4em]
1, & \begin{array}[t]{l}
\text{if } \Delta^*_{01}(\alpha^m,\alpha^b)>0\text{ and}\\
\Delta^*_{21}(\alpha^m,\alpha^b)\ge 0,
\end{array}\\[0.4em]
2, & \text{otherwise.}
\end{cases}
\]
This rule is optimal. Indeed, if $\Delta^*_{01}\le 0$ and $\Delta^*_{02}\le 0$, then $Q_0^*\le Q_1^*$ and $Q_0^*\le Q_2^*$, so action $0$ is optimal. If $\Delta^*_{01}>0$ and $\Delta^*_{21}\ge 0$, then $Q_1^*<Q_0^*$ and $Q_1^*\le Q_2^*$, so action $1$ is optimal. In all remaining cases action $2$ is optimal: if $\Delta^*_{01}\le 0$ and the first case fails, then necessarily $\Delta^*_{02}>0$, hence $Q_2^*<Q_0^*\le Q_1^*$; if $\Delta^*_{01}>0$ and the second case fails, then necessarily $\Delta^*_{21}<0$, hence $Q_2^*<Q_1^*<Q_0^*$.

Now fix $\alpha^b\in\mathbb N$ and define
\begin{align*}
\mathcal A_0(\alpha^b)
&:=\{\alpha^m\ge \alpha^b:\ \Delta^*_{01}(\alpha^m,\alpha^b)\le 0,\\
&\qquad \Delta^*_{02}(\alpha^m,\alpha^b)\le 0\},\\
\mathcal A_1(\alpha^b)
&:=\{\alpha^m\ge \alpha^b:\ \Delta^*_{01}(\alpha^m,\alpha^b)>0,\\
&\qquad \Delta^*_{21}(\alpha^m,\alpha^b)\ge 0\}.
\end{align*}
By construction, action $0$ is optimal on $\mathcal A_0(\alpha^b)$, action $1$ is optimal on
$\mathcal A_1(\alpha^b)$, and action $2$ is optimal on the complement of
$\mathcal A_0(\alpha^b)\cup \mathcal A_1(\alpha^b)$ in $\{\alpha^m\ge \alpha^b\}$.

By Lemma \ref{lemma:ordered-cost-diagonal-zero}, $\alpha^b\in\mathcal A_0(\alpha^b)$, so $\mathcal A_0(\alpha^b)$ is nonempty. Moreover, by Lemma \ref{lemma:delta-non-decreasing-in-a_s}, both $\Delta^*_{01}$ and $\Delta^*_{02}$ are nondecreasing in $\alpha^m$. Hence, if $\alpha^m\in\mathcal A_0(\alpha^b)$ and $\tilde\alpha^m$ satisfies $\alpha^b\le \tilde\alpha^m\le \alpha^m$, then
\begin{align*}
\Delta^*_{01}(\tilde\alpha^m,\alpha^b)
&\le \Delta^*_{01}(\alpha^m,\alpha^b)\le 0,\\
\Delta^*_{02}(\tilde\alpha^m,\alpha^b)
&\le \Delta^*_{02}(\alpha^m,\alpha^b)\le 0,
\end{align*}
so $\tilde\alpha^m\in\mathcal A_0(\alpha^b)$. Therefore $\mathcal A_0(\alpha^b)$ is an initial
segment of $\{\alpha^m\ge \alpha^b\}$.

Similarly, by Lemma \ref{lemma:delta-non-decreasing-in-a_s}, both
$\Delta^*_{01}$ and $\Delta^*_{21}$ are nondecreasing in
$\alpha^m$. Hence, if $\alpha^m\in\mathcal A_1(\alpha^b)$ and $\tilde\alpha^m\ge \alpha^m$, then
\begin{align*}
\Delta^*_{01}(\tilde\alpha^m,\alpha^b)
&\ge \Delta^*_{01}(\alpha^m,\alpha^b)>0,\\
\Delta^*_{21}(\tilde\alpha^m,\alpha^b)
&\ge \Delta^*_{21}(\alpha^m,\alpha^b)\ge 0,
\end{align*}
so $\tilde\alpha^m\in\mathcal A_1(\alpha^b)$. Therefore $\mathcal A_1(\alpha^b)$ is a terminal
segment of $\{\alpha^m\ge \alpha^b\}$.

We may thus define
\[
\tau_1(\alpha^b):=\sup \mathcal A_0(\alpha^b)\in \mathbb N_0\cup\{+\infty\},
\]
\[
\tau_2(\alpha^b):=
\begin{cases}
\inf \mathcal A_1(\alpha^b)-1, & \text{if } \mathcal A_1(\alpha^b)\neq\emptyset,\\
+\infty, & \text{if } \mathcal A_1(\alpha^b)=\emptyset.
\end{cases}
\]
Since $\mathcal A_0(\alpha^b)$ is an initial segment and $\mathcal A_1(\alpha^b)$ is a terminal
segment, we have
\[
\mathcal A_0(\alpha^b)=\{\alpha^m\ge \alpha^b:\alpha^m\le \tau_1(\alpha^b)\}
\]
\[
\mathcal A_1(\alpha^b)=\{\alpha^m\ge \alpha^b:\alpha^m>\tau_2(\alpha^b)\}.
\]
Because the two sets are disjoint, $\tau_1(\alpha^b)\le \tau_2(\alpha^b)$. Hence
\[
u^*(\alpha^m,\alpha^b)=
\begin{cases}
0, & \text{if } \alpha^m\le \tau_1(\alpha^b),\\
2, & \text{if } \tau_1(\alpha^b)<\alpha^m\le \tau_2(\alpha^b),\\
1, & \text{if } \alpha^m>\tau_2(\alpha^b).
\end{cases}
\]

It remains to prove that $\tau_1$ is nondecreasing in $\alpha^b$. Let
$\alpha^m\in\mathcal A_0(\alpha^b)$ with $\alpha^m\ge \alpha^b+1$. By
Lemma \ref{lemma:delta-non-increasing-in-a_b},
\begin{align*}
\Delta^*_{01}(\alpha^m,\alpha^b+1)
&\le \Delta^*_{01}(\alpha^m,\alpha^b)\le 0,\\
\Delta^*_{02}(\alpha^m,\alpha^b+1)
&\le \Delta^*_{02}(\alpha^m,\alpha^b)\le 0,
\end{align*}
so $\alpha^m\in\mathcal A_0(\alpha^b+1)$. Therefore
\[
\mathcal A_0(\alpha^b)\setminus\{\alpha^b\}\subseteq \mathcal A_0(\alpha^b+1).
\]
Since Lemma \ref{lemma:ordered-cost-diagonal-zero} gives
$\alpha^b+1\in\mathcal A_0(\alpha^b+1)$, and both sets are initial segments, it follows that
\[
\tau_1(\alpha^b)\le \tau_1(\alpha^b+1).
\]
Thus $\tau_1$ is nondecreasing in $\alpha^b$.
\end{proof}

\subsection{Truncation of the State Set}

Since the state space of the single-source MDP is unbounded, the optimal solution requires to solve the Bellman equation on an infinite amount of states. In practical implementations, this is clearly an impossible task. For this reason, one must truncate the state space, and solve the resulting truncated MDP. The truncated model then works as follows:
\begin{itemize}
    \item Offline, the Bellman equation is modified at the boundary: whenever a transition would lead outside the truncated state space, the corresponding value function is replaced by its clipped boundary counterpart.
    \item Online, whenever the state exceeds the truncation value, it gets projected back to the boundary.
\end{itemize}
In this way, the problem in \eqref{eq:problem-formulation} can be approximated by a finite-state MDP, which is solvable. However, such truncation of the model may lead to a policy that is different from the optimal policy of the original unbounded problem.

In this section, we explicitly quantify the error induced by truncation and derive a possible closed-form criterion for selecting the truncation level. This shows that the truncated model maintains a controlled approximation of the original MDP, and provides a justification for the numerical procedure used to compute the optimal policy. Precisely, the goal of the following analysis is to bound the value error introduced by the clipped approximation of the original unbounded dynamic program.

For a fixed $A\in\mathbb N$, define the truncated state space
\[
(\alpha^m,\alpha^b)\in \mathcal S_A := \{(\alpha^m,\alpha^b): 1\leq \alpha^b\leq \alpha^m\leq A\}.
\]
Whenever the state exceeds the boundary in the truncated model, it is clipped to $A$. Let $d^\pi(s)$ be the discounted occupancy measure of state $s$ under a policy $\pi$, namely
\[
d^\pi(s):=(1-\gamma)\sum_{k=0}^{\infty}\gamma^k \Pr^\pi\!\left(S_k=s\mid S_0=(1,1)\right).
\]
More generally, for any subset $B$ of the state space, let $d^\pi(B):=\sum_{s\in B} d^\pi(s)$. For every $i\in \mathbb N$, define $B_i:=\{(\alpha^m,\alpha^b): \alpha^m\geq A+i\}$.

\begin{lemma}\label{lem:outer-region-discounted-occupancy}
For every $i\in \mathbb N$, the discounted occupancy measure satisfies:
\[
d^\pi(B_i)\leq \gamma^{A+i-1}.
\]
\end{lemma}

\begin{proof}
From the state transition in \eqref{eq:state-evolution} $\alpha^m$ can increase by at most one unit per slot under any action. Therefore, reaching $B_i$ from $(1,1)$ requires at least $A+i-1$ time steps. It follows that
\begin{align*}
d^\pi(B_i)
&=(1-\gamma)\sum_{k=0}^{\infty}\gamma^k
\Pr^\pi\bigl(S_k\in B_i\mid S_0=(1,1)\bigr)\\
&\leq (1-\gamma)\sum_{k=A+i-1}^{\infty}\gamma^k
=\gamma^{A+i-1}.
\end{align*}
\end{proof}

Now let $V^\infty_{\pi^*}(1,1)$ be the value of the unbounded MDP at $(1,1)$, and let $V^A_{\pi_A}(1,1)$ be the value of the truncated MDP at the same initial state, where $\pi^*$ and $\pi_A$ are optimal for the unbounded and truncated models, respectively.

\begin{theorem}\label{th:truncation-error}
The error due to the truncation satisfies
\[
V^\infty_{\pi^*}(1,1)-V^A_{\pi_A}(1,1)\leq \frac{\gamma^A}{(1-\gamma)^2}.
\]
\end{theorem}

\begin{proof}
Let $\pi_A^{\mathrm{ext}}$ be the extension of $\pi_A$ to the unbounded model defined by
\[
\pi_A^{\mathrm{ext}}(\alpha^m,\alpha^b)=\pi_A\!\left(\min\{\alpha^m,A\},\min\{\alpha^b,A\}\right).
\]
Since $\pi_A^{\mathrm{ext}}$ is feasible for the unbounded MDP,
\[
V^\infty_{\pi^*}(1,1)\leq V^\infty_{\pi_A^{\mathrm{ext}}}(1,1).
\]
Coupling the unbounded and truncated processes under the same realization of the channel outcomes,
\begin{align*}
V^\infty_{\pi^*}(1,1)-V^A_{\pi_A}(1,1)
&\leq V^\infty_{\pi_A^{\mathrm{ext}}}(1,1)-V^A_{\pi_A}(1,1)\\
&=\mathbb E\Biggl[\sum_{k=0}^{\infty}\gamma^k
\bigl(\alpha_k^m-\min\{\alpha_k^m,A\}\bigr)\\
&\qquad\Big|\,S_0=(1,1)\Biggr]\\
&=\frac{1}{1-\gamma}
\sum_{s\notin \mathcal S_A} d^{\pi_A^{\mathrm{ext}}}(s)(\alpha^m-A).
\end{align*}
Every reachable state satisfies $\alpha^m\geq \alpha^b$, so $s\notin \mathcal S_A$ implies $\alpha^m>A$. Hence
\[
\alpha^m-A=\sum_{i=1}^{\infty}\mathds{1}\{\alpha^m\geq A+i\}.
\]
Substituting into the previous expression and using Lemma \ref{lem:outer-region-discounted-occupancy},
\begin{align*}
V^\infty_{\pi^*}(1,1)-V^A_{\pi_A}(1,1)
&\leq \frac{1}{1-\gamma}\sum_{i=1}^{\infty}
 d^{\pi_A^{\mathrm{ext}}}(B_i)\\
&\leq \frac{1}{1-\gamma}\sum_{i=1}^{\infty}\gamma^{A+i-1}
=\frac{\gamma^A}{(1-\gamma)^2}.
\end{align*}
\end{proof}

This closed-form bound can be directly inverted to choose the truncation level for a given tolerance. If one requires the truncation error to be at most $\varepsilon$, it is enough to impose $\gamma^A/(1-\gamma)^2\leq \varepsilon$, which leads to
\begin{equation}\label{eq:truncation-threshold}
A\geq \left\lceil\frac{\log\bigl(\varepsilon(1-\gamma)^2\bigr)}{\log(\gamma)}\right\rceil.
\end{equation}

The value $\varepsilon$ is a tolerance on the differencial total discounted cost. A more intuitive value is $\hat\varepsilon:=(1-\gamma)\varepsilon$, which corresponds to the tolerance on the differencial discounted cost per time slot. In other words, $\hat\varepsilon$ is the error that, if experienced for every time slot, would lead to a total differencial discounted cost equal to $\varepsilon$. Note that $\hat\varepsilon$ has the same scale as the cost function $g(S,u)$, so it has a clear interpretation. The criterion \eqref{eq:truncation-threshold} becomes
\begin{equation}\label{eq:truncation-threshold-2}
A\geq \left\lceil\frac{\log\bigl(\hat\varepsilon(1-\gamma)\bigr)}{\log(\gamma)}\right\rceil.
\end{equation}
Table \ref{tab:truncation_levels_normalized} reports the minimum values of $A$ satisfying \eqref{eq:truncation-threshold-2} for different values of $\gamma$ and $\hat\varepsilon$.

\begin{table}[ht]
\centering
\caption{Minimum truncation level $A$ satisfying \eqref{eq:truncation-threshold-2}.}
\label{tab:truncation_levels_normalized}
\resizebox{\linewidth}{!}{%
\begin{tabular}{c|ccccc}
\hline
$\hat\epsilon$ & $\gamma=0.50$ & $\gamma=0.70$ & $\gamma=0.85$ & $\gamma=0.90$ & $\gamma=0.95$ \\
\hline
$1$             & 1  & 4  & 12 & 22 & 59  \\
$5\cdot 10^{-1}$& 2  & 6  & 16 & 29 & 72  \\
$10^{-1}$       & 5  & 10 & 26 & 44 & 104 \\
$5\cdot 10^{-2}$& 6  & 12 & 31 & 51 & 117 \\
$10^{-2}$       & 8  & 17 & 41 & 66 & 149 \\
$5\cdot 10^{-3}$& 9  & 19 & 45 & 73 & 162 \\
$10^{-3}$       & 11 & 23 & 55 & 88 & 194 \\
\hline
\end{tabular}%
}
\end{table}

\section{Multiple process-monitor pairs Scenario}\label{sec:multi-source-scenario}

In this section, we study a constrained multiple process-monitor pairs scheduling problem in which a single base station must share its sensing and communication capability among several independent monitor-source pairs. In contrast to the single process-monitor case, the decision process can no longer be treated separately for each source, since the base station can actively serve only a limited number of subsystems at each time slot. The controller must therefore decide, at every time step, which subsystems should be addressed and which ISAC action should be applied to each selected subsystem. This coupling across subsystems makes the exact dynamic programming solution intractable when the number of sources grows.

To model this setting, we extend the single process-monitor formulation by introducing an idle action, which represents the decision not to address a given subsystem in a given slot. This leads naturally to a constrained restless multi-armed bandit (RMAB) formulation. The key difficulty is that every subsystem (arm) keeps evolving over time, including those that are not selected, so the overall state process remains coupled through the activation constraint. A standard way to handle this difficulty is through a Lagrangian relaxation, which decouples the global scheduling problem into a family of relaxed single-arm problems. This relaxation is the basis for Whittle index policies, which assign to each subsystem a priority value and then activate the subsystems with the largest priorities. In this way, a high-dimensional scheduling problem is replaced by an offline index computation and a simple online ranking rule.

Compared with standard AoI scheduling models, the relaxed single-arm problem considered here has two distinctive features. First, each subsystem state contains two coupled freshness variables, corresponding to the AoI at the monitor and at the base station. Second, each active arm has three possible active modes, with different state transitions, success probabilities, and operational costs. Therefore, the resulting index construction must account not only for whether a source should be scheduled, but also for which sensing-communication action should be selected once the source is activated.

Our analysis focuses on two regimes. First, we identify a sufficient condition under which the relaxed single-arm problem is Whittle-indexable, so that an exact Whittle index policy can be defined. Second, since this sufficient condition need not hold for all parameter values, we also develop an approximate Whittle index policy that remains computationally light and can still be used beyond the guaranteed indexable regime.

The rest of the section is organized as follows. We first formulate the multi-source problem as an RMAB and derive the relaxed single-arm problem. We then study indexability and the exact Whittle index policy. Finally, we introduce an approximate index policy for the non-indexable regime.

\subsection{RMAB Formulation}
	
We consider a multi-source scenario with $N$ total subsystems, each of which can be modeled as a single process-monitor MDP as presented in Section \ref{sec:single-source-scenario}. Every subsystem $i$ has its own parameters: $c^i_0$, $c^i_1$, $c^i_2$, $\lambda^i_0$, $\lambda^i_1$, $\lambda^i_2$, and $A^i$ for numerical implementation. A constraint forces the base station to address at most $M<N$ sources, each with its own optimal action. To implement this constraint, we now introduce action $u^i=3$, which means that the subsystem $i$ stays idle. When a subsystem stays idle, its state variables both increase by $1$. The optimal policy for the base station is now a vector $\bm{\pi}=\{\pi^i\}_{i=1}^N$. Furthermore, let $d_{\pi^i}(k) := 1$ if $\pi^i(k)\in\{0,1,2\}$ and $0$ if $\pi^i(k)=3$ for subsystem $i$ and for time step $k$, i.e. $d_{\pi^i}(k) = 1$ if the base station addresses source $i$ at time step $k$. The RMAB can be formulated as follows:
\begin{equation}\label{eq:multi-source-problem}
\begin{aligned}
&\min_{\boldsymbol{\pi}}\mathbb{E}\Bigr[\sum_{k=0}^{\infty}\gamma^k\sum_{i=1}^Ng(S^i_k,u^i_k)\Bigl]\\&
\text{s.t. } \sum_{i=1}^Nd_{\pi^i}(k)\leq M,\quad\forall k \geq 0,
\end{aligned}
\end{equation}
where
\[
\begin{aligned}
g(S^i_k,u^i_k)&=\alpha^{m,i}_k+c^i_0\mathds{1}\{u^i_k=0\}\\
&\quad+c^i_1\mathds{1}\{u^i_k=1\}+c^i_2\mathds{1}\{u^i_k=2\}.
\end{aligned}
\]

\subsection{Lagrangian Relaxation}
The multiple process-monitor pairs problem in \eqref{eq:multi-source-problem} is computationally demanding to solve. To obtain a scheduling rule, we use a Lagrangian relaxation, which is the standard starting point for Whittle index methods. The idea is to replace the activation constraint with a scalar term in the objective function. In this way, the constraint is no longer enforced for every time step $k$; instead, one introduces a multiplier $W$ that measures the value of leaving a subsystem idle. This relaxed formulation is useful because it turns the original coupled problem into a collection of local control problems, one for each subsystem.

Imposing a Lagrange multiplier $W$, the optimization problem in \eqref{eq:multi-source-problem} can be written as:
\begin{equation}\label{eq:relaxed-problem}
\min_{\boldsymbol{\pi}}\mathbb{E}\Biggl[
\sum_{k=0}^{\infty}\gamma^k\sum_{i=1}^N
\bigl(g(S^i_k,u^i_k)-W\mathds{1}\{u^i_k=3\}\bigr)
\Biggr]
\end{equation}

The Lagrangian relaxation decouples the constrained RMAB into $N$ independent discounted dynamic programs, one for each subsystem $i$. Let $s^{+1}$ denote the state whose coordinates are those in state $s$ increased by $1$, and let $s'_u$ denote the arrival state after a successful realization of action $u$ in state $s$. For a fixed subsystem $i$, the relaxed Bellman equation is
\[
V(s^i,W)=\min_{u\in\{0,1,2,3\}}Q_u(s^i,W),
\]
where
\[
Q_3(s^i,W)=g(s^i,3)-W+\gamma V(s^{i,+1},W),
\]
and, for $u\in\{0,1,2\}$,
\[
\begin{aligned}
Q_u(s^i,W)&=g(s^i,u)+\gamma\lambda^i_uV(s^{\prime i}_u,W)\\
&\quad+\gamma(1-\lambda^i_u)V(s^{i,+1},W).
\end{aligned}
\]
Accordingly, we define the active-idle difference functions
\begin{equation}
\Delta^i_{u3}(s^i,W):=Q_u(s^i,W)-Q_3(s^i,W),\qquad u\in\{0,1,2\}.
\end{equation}
By direct substitution,
\begin{subequations}\label{eq:deltas-relaxed}
\begin{align}
\Delta^i_{03}(s^i,W)&=c^i_0+W-\gamma\lambda^i_0\bigl[V(s^{i,+1},W)-V(s^{\prime i}_0,W)\bigr],\\
\Delta^i_{13}(s^i,W)&=c^i_1+W-\gamma\lambda^i_1\bigl[V(s^{i,+1},W)-V(s^{\prime i}_1,W)\bigr],\\
\Delta^i_{23}(s^i,W)&=c^i_2+W-\gamma\lambda^i_2\bigl[V(s^{i,+1},W)-V(s^{\prime i}_2,W)\bigr].
\end{align}
\end{subequations}
These quantities compare each active action with the idle one in the relaxed problem. In particular, if all three of them are nonnegative at a given state, then staying idle is optimal for that state. Let $P^i(W)$ be the set of states for which the optimal action is $u=\mathsf{idle}$, i.e.
\[
P^i(W):=\{s^i:u^{*,i}(s^i)=3\}.
\]
Equivalently,
\[
\begin{aligned}
P^i(W)=\{s^i:&\Delta^i_{03}(s^i,W)\geq 0,
\Delta^i_{13}(s^i,W)\geq 0,\\
&\Delta^i_{23}(s^i,W)\geq 0\}.
\end{aligned}
\]

The notion of indexability is central in Whittle theory. A subsystem is said to be indexable if increasing the idle subsidy $W$ can only enlarge the set of states in which staying idle is optimal. In other words, as passivity becomes more attractive, the system should move monotonically toward the idle action. The problem in \eqref{eq:multi-source-problem} is indexable if, for any $W'>W$, $P^i(W)\subseteq P^i(W')$.

If the problem is indexable, then the Whittle index is well defined: for a given state, it is the minimum value of the subsidy $W$ for which action $\mathsf{idle}$ becomes optimal in that state. In mathematical terms, the Whittle index of state $s$ in subsystem $i$ is defined as
\begin{equation}\label{eq:whittle-index}
W^i(s^i):=\inf\{W:s^i\in P^i(W)\}.
\end{equation}

This index can be interpreted as a state-dependent priority value for each subsystem, and it directly induces the Whittle index policy. In practice, the Whittle index of every state of every subsystem is computed offline and stored. Then, online, at each time step, the base station observes the current state of each subsystem, retrieves the corresponding index values, and ranks the subsystems in decreasing order of priority. The base station then activates the $M$ subsystems with the largest indices and applies to each of them the corresponding local optimal action of the relaxed single-arm problem, while the remaining $N-M$ subsystems stay idle. In this way, the original high-dimensional scheduling problem is replaced by an offline index computation and a simple online lookup-and-sorting procedure.

We now state a closed-form sufficient condition under which subsystem $i$ is indexable, so that the Whittle index policy can be used.
\begin{theorem}\label{th:indexability}
If
\[
\gamma\leq\frac{1}{1+\lambda^i_1},
\]
then subsystem $i$ is indexable.
\end{theorem}
\begin{proof}
Fix $W'>W$. Since only the idle action depends explicitly on $W$, the relaxed value function is nonincreasing in $W$ and satisfies
\[
0\leq V(s,W)-V(s,W')\leq \frac{W'-W}{1-\gamma},\qquad \forall s.
\]
Using \eqref{eq:deltas-relaxed}, for any $u\in\{0,1,2\}$ we obtain
\begin{align*}
\Delta^i_{u3}&(s^i,W')-\Delta^i_{u3}(s^i,W)\\
&=(W'-W)+\gamma\lambda^i_u\Bigl[\\
&\qquad \bigl(V(s^{i,+1},W)-V(s^{i,+1},W')\bigr)\\
&\qquad-\bigl(V(s^{\prime i}_u,W)-V(s^{\prime i}_u,W')\bigr)\Bigr]\\
&\geq (W'-W)\\
&\quad-\gamma\lambda^i_u\left\|V(\cdot,W)-V(\cdot,W')\right\|_\infty.
\end{align*}
An upper bound for the difference $\left\|V(\cdot,W)-V(\cdot,W')\right\|_\infty$ is given by \cite[Theorem 12]{csaji2008value_function}:
\begin{equation}\label{eq:V-difference-bound}
\left\|V(\cdot,W)-V(\cdot,W')\right\|_\infty \leq \frac{\bigl|W-W'\bigr|}{1-\gamma}.
\end{equation}
Therefore,
\[
\begin{aligned}
\Delta^i_{u3}(s^i,W')-\Delta^i_{u3}(s^i,W)
&\geq (W'-W)\left(1-\frac{\gamma\lambda^i_u}{1-\gamma}\right).
\end{aligned}
\]

Since $\lambda^i_u\leq \lambda^i_1$ for all $u\in\{0,1,2\}$ and $\gamma\leq 1/(1+\lambda^i_1)$, the right-hand side is nonnegative. Hence each $\Delta^i_{u3}(s^i,W)$ is nondecreasing in $W$. Consequently, if $s^i\in P^i(W)$, then $\Delta^i_{u3}(s^i,W')\geq 0$ for all $u\in\{0,1,2\}$, which proves that $P^i(W)\subseteq P^i(W')$.
\end{proof}

Now let $s^{+1}$ denote the state whose coordinates are those in state $s$ increased by $1$, and let $s'_u$ denote the arrival state after a successful realization of action $u$ in state $s$.
At $W=W^i(s^i)$, we have
\[
\min_{u\in\{0,1,2\}}\Delta^i_{u3}(s^i,W^i(s^i))=0,
\]
which directly yields the following characterization:
\begin{equation}\label{eq:whittle-bellman}
\begin{aligned}
W^i(s^i)=\max_{u^i\in\{0,1,2\}}\Bigl[&
\gamma\lambda^i_uV\bigl(s^{i,+1},W^i(s^i)\bigr)\\
&-\gamma\lambda^i_uV\bigl(s^{\prime i}_u,W^i(s^i)\bigr)-c^i_u\Bigr].
\end{aligned}
\end{equation}

\subsection{Whittle Index Policy Under Indexable Regime}

We now focus on the regime in which the sufficient condition of Theorem \ref{th:indexability} is satisfied, so that the Whittle index is well defined as in \eqref{eq:whittle-index}. In this case, the relaxed single-arm problem admits an exact Whittle index for every state, which can be used as a priority value for scheduling.

Under the indexable regime, the Whittle index of each state can be numerically computed through the bisection method reported in Algorithm \ref{alg:whittle-exact}.

\begin{algorithm}[ht]
\caption{Whittle index computation}
\label{alg:whittle-exact}
\begin{algorithmic}[1]
    \For{all states $s^i$ in the state space}
		\State Choose $W_{\min}$ and $W_{\max}$ such that
		$s^i \notin P^i(W_{\min})$ and $s^i \in P^i(W_{\max})$
		\State Initialization:
		$\underline{W}^{(0)}(s^i) \leftarrow W_{\min}$,
		$\overline{W}^{(0)}(s^i) \leftarrow W_{\max}$
		\For{$k = 1, \ldots, K_{\max}$}
		\State Set
		\[
		W^{(k)}(s^i) \leftarrow
		\frac{\underline{W}^{(k-1)}(s^i)+\overline{W}^{(k-1)}(s^i)}{2}
		\]
		\State Solve the local dynamic program using
		$W^i = W^{(k)}(s^i)$
		\If{$s^i \in P^i(W^{(k)}(s^i))$}
		\State $\underline{W}^{(k)}(s^i) \leftarrow \underline{W}^{(k-1)}(s^i)$
		\State $\overline{W}^{(k)}(s^i) \leftarrow W^{(k)}(s^i)$
		\Else
		\State $\underline{W}^{(k)}(s^i) \leftarrow W^{(k)}(s^i)$
		\State $\overline{W}^{(k)}(s^i) \leftarrow \overline{W}^{(k-1)}(s^i)$
		\EndIf
		\If{$\left|\overline{W}^{(k)}(s^i)-\underline{W}^{(k)}(s^i)\right| < \varepsilon$}
		\State break
		\EndIf
		\EndFor
		\State Compute
		\[
		W(s^i)=\frac{\underline{W}^{(k)}(s^i)+\overline{W}^{(k)}(s^i)}{2}
		\]
		\EndFor
\end{algorithmic}
\end{algorithm}
	
The output of Algorithm \ref{alg:whittle-exact} will be the Whittle index of each state $s^i$ of subsystem $i$, up to a certain tolerance $\varepsilon$. The process must be repeated for every $i\in\{1,2,\dots,N\}$. Online, the implementation of the Whittle index policy is done through Algorithm \ref{alg:whittle-online}:

\begin{algorithm}[ht]
\caption{Implementation of the Whittle index policy}
\label{alg:whittle-online}
\begin{algorithmic}[1]
    \For{$k=0,1,2,\dots$}
        \State Observe the current state $s_k^i$ of each subsystem $i\in\{1,2,\dots,N\}$
        \State Retrieve the Whittle indices $\{W^i(s_k^i)\}_{i=1}^N$
        \State Sort the indices $W^i(s_k^i)$ in decreasing order
        \State Let $\mathcal I_k$ be the set of the $M$ subsystems with the largest indices
        \ForAll{$i\in\mathcal I_k$}
            \State Select the active action
            \[
            u_k^i \in \arg\min_{u\in\{0,1,2\}} Q_u^i\bigl(s_k^i, W^i(s_k^i)\bigr)
            \]
        \EndFor
        \ForAll{$i\notin\mathcal I_k$}
            \State Set $u_k^i \leftarrow 3$
        \EndFor
        \State Apply the action vector $u_k=(u_k^1,\dots,u_k^N)$
    \EndFor
\end{algorithmic}
\end{algorithm}

\subsection{Approximate Whittle Index Policy Beyond the Indexable Regime}\label{sec:AWIP}

The sufficient condition for indexability presented above can be restrictive. Moreover, having a two-dimensional state space, calculating the Whittle index for all states is generally challenging.

In this section, we propose a heuristic, well-performing policy that assigns an index $\widetilde{W}(s^i)$ to every state through linear interpolation. More precisely, we first compute the index values through Algorithm \ref{alg:whittle-exact} on a set of anchor states, namely the diagonal states and the boundary states, and then use these values to interpolate the index over the interior of the state space. The interpolated value is then used to solve the local dynamic program only once, which yields the final approximate index $\widetilde{W}(s^i)$. In this way, we obtain a simple index-based policy that can be computed over the whole state space and used for scheduling in a broad range of parameter regimes. This makes the policy attractive both when the sufficient indexability condition is not satisfied and, more generally, when a much faster computational procedure is needed. The offline computation method for the approximate indices is described in Algorithm \ref{alg:whittle-boundary-interp}.

Online, similarly to the Whittle index policy, for every time step the base station addresses the first $M$ subsystems with the highest $\widetilde{W}(s^i)$, applying Algorithm \ref{alg:whittle-online} to the approximate indices.

\begin{algorithm}[ht]
\caption{Approximate Whittle index computation}
\label{alg:whittle-boundary-interp}
\begin{algorithmic}[1]
	\State Define
		\[
\begin{aligned}
\mathcal D^i&:=\{(a,a): a=1,\dots,A^i\},\\
\mathcal B^i&:=\{(A^i,b): b=1,\dots,A^i\}.
\end{aligned}
\]
			
		\ForAll{$s^i \in \mathcal D^i \cup \mathcal B^i$}
			\State Find $\bar W^i(s^i)$ through Algorithm \ref{alg:whittle-exact}
			\State Set $\widetilde{W}^i(s^i)\gets \bar W^i(s^i)$
			\EndFor
			
			\For{$b=1,\dots,A^i-1$}
			\For{$a=b+1,\dots,A^i-1$}
			\State Compute the linear interpolation
			\[
			\widehat{W}^i(s^i)\gets\widetilde{W}^i(b,b)+\frac{a-b}{A^i-b}\Big(
			\widetilde{W}^i(A^i,b)-\widetilde{W}^i(b,b)\Big)
			\]
			\State Solve the local dynamic program using $W^i=\widehat W^i(s^i)$
			\State Compute
			\[
\begin{aligned}
\widetilde{W}^i(s^i)=\max_{u^i \in \{0,1,2\}}\Big[&
\gamma\lambda^i_u V\bigl(s^{i,+1},\widehat{W}^i(s^i)\bigr)\\
&-\gamma\lambda^i_u V\bigl(s^{\prime i}_u, \widehat W^i(s^i)\bigr)-c^i_u\Big]
\end{aligned}
\]
			\EndFor
		\EndFor
\end{algorithmic}
\end{algorithm}

The output of Algorithm \ref{alg:whittle-boundary-interp} is an index approximation for each state $s^i$ of subsystem $i$, and the process must be repeated for every $i\in\{1,2,\dots,N\}$. This heuristic solution computes approximate Whittle indices for the boundary states $(A^i,\alpha^{b,i})$ and $(\alpha,\alpha)$ through Algorithm \ref{alg:whittle-exact}, assigns values $\widehat{W}(s^i)$ to all other states by linear interpolation, and then updates them by solving the corresponding dynamic program only once.

Whenever the exact Whittle index is well defined, the index assigned to any boundary state through Algorithm \ref{alg:whittle-exact} coincides with the real Whittle index; otherwise, it can be viewed as a heuristic subsidy score. For the inner states, $\widetilde{W}(s^i)$ are also heuristic priority scores. This avoids the full iterative procedure needed to compute the Whittle index over the whole state space, and also provides a practical method when indexability is not guaranteed. In the indexable case, Theorem \ref{th:UB-for-whittle-error} quantifies the sensitivity of this approximation to the interpolation error.

Let $\lambda^i_{\max}:=\max_{u^i\in\{0,1,2\}}\{\lambda_u^i\}, \ \forall i\in\{1,2,\dots,N\}.$
	
\begin{theorem}\label{th:UB-for-whittle-error}
Let the subsystem $i$ be indexable. With reference to Algorithms \ref{alg:whittle-exact} and \ref{alg:whittle-boundary-interp}, for the subsystem $i$, the approximation error satisfies:
\[
\bigl|\widetilde{W}^i(s^i)-W^i(s^i)\bigr|
\leq\frac{2\gamma\lambda^i_{\max}}{1-\gamma}
\bigl|\widehat{W}^i(s^i)-W^i(s^i)\bigr|.
\]
\end{theorem}
	
\begin{proof}
We refer to subsystem $i$, and we omit index $i$ to simplify notations.

From \eqref{eq:whittle-bellman} and from the definition of $\widetilde{W}(s)-W(s)$ in our heuristic policy,
\[
\begin{aligned}
\bigl|\widetilde{W}(s)-W(s)\bigr|=\bigl|\max_{u\in\{0,1,2\}}\Bigl[
&\gamma\lambda_u V\bigl(s^{+1}, \widehat W(s)\bigr)\\
&-\gamma\lambda_u V\bigl(s'_u, \widehat W(s)\bigr)-c_u\Bigr]\\
&-\max_{u\in\{0,1,2\}}\Bigl[\gamma\lambda_uV(s^{+1},W(s))\\&-\gamma\lambda_uV(s'_u,W(s))-c_u\Bigr]\bigr|.
\end{aligned}
\]
By the triangular inequality we can write
\[
\begin{aligned}
\bigl|\widetilde{W}(s)-W(s)\bigr|
&\leq \gamma\max_{u\in\{0,1,2\}}\lambda_u
\Bigl|\\
&\qquad V\bigl(s^{+1},\widehat{W}(s)\bigr)-V\bigl(s^{+1},W(s)\bigr)\\
&\qquad + V\bigl(s'_u,W(s)\bigr)-V\bigl(s'_u,\widehat{W}(s)\bigr)\Bigr|\\
&\leq \gamma\max_{u\in\{0,1,2\}}\lambda_u
\Bigl[\\
&\qquad \bigl|V\bigl(s^{+1},\widehat{W}(s)\bigr)-V\bigl(s^{+1},W(s)\bigr)\bigr|\\
&\qquad + \bigl|V\bigl(s'_u,W(s)\bigr)-V\bigl(s'_u,\widehat{W}(s)\bigr)\bigr|\Bigr].
\end{aligned}
\]
An upper bound for both differences in the value functions is again given by \eqref{eq:V-difference-bound}:
\[
\left\|V\bigl(s,\widehat{W}(s)\bigr)-V\bigl(s,W(s)\bigr)\right\|_\infty
\leq\frac{\bigl|\widehat{W}(s)-W(s)\bigr|}{1-\gamma}.
\]
Therefore,
\[
\begin{aligned}
\bigl|\widetilde{W}(s)-W(s)\bigr|&\leq\;2\gamma\lambda_{\max}\left\|V\bigl(\widehat{W}(s)\bigr)-V\bigl(W(s)\bigr)\right\|_\infty\\
&\leq\frac{2\gamma\lambda_{\max}}{1-\gamma}\,\bigr|\widehat{W}(s)-W(s)\bigl|.
\end{aligned}
\]
\end{proof}

\begin{remark}

Theorem \ref{th:UB-for-whittle-error} can be leveraged to get a more formal upper bound for the interpolation error. Let $H^{i,b}$ be the function defined by:
\[
H^{i,b}:=\max_{a=b,\dots,A_i-2}\bigl|W^i(a+2,b)-2W^i(a+1,b)+W^i(a,b)\bigr|
\]
for every truncated subsystem $i$. Since the row $\{(a,b):a=b,\dots,A^i\}$ is finite, $H^{i,b}$ is well defined.
Now consider the interpolation
\[
\begin{aligned}
\widehat W^i(a,b)&=W^i(b,b)+\frac{a-b}{A^i-b}
\bigl(W^i(A^i,b)-W^i(b,b)\bigr),\\
&\hspace{2em} b\leq a\leq A^i.
\end{aligned}
\]
One can show that
\[\bigl|\widehat W^i(a,b)-W^i(a,b)\bigr|\leq\frac{H^{i,b}}{2}(a-b)(A^i-a).
\]
Hence, by Theorem \ref{th:UB-for-whittle-error}, the approximation produced by one Whittle update satisfies
\begin{equation}\label{eq:H-bound}
\bigl|\widetilde{W}^i(a,b)-W^i(a,b)\bigr|\leq\frac{\gamma\lambda^i_{\max}}{1-\gamma}
\,H^{i,b}\,(a-b)(A^i-a).
\end{equation}

The bound defined in \eqref{eq:H-bound} is zero at the states $(b,b)$ and $(A_i,b)$, and is maximized at the midpoint of the interval $[b, A_i]$, since the factor $(a-b)(A^i-a)$ is a concave quadratic function of $a$.

\end{remark}

One may want to apply this remark to optimize their heuristic solution a posteriori. For example, if the numerically computed $H^{i,b}$ has a high curvature for a certain $b$, splitting the interval $[b, A_i]$ for interpolation is the theoretically optimal way to minimize the error upper bound, and might lead to significant improvements.

\section{Numerical Results}\label{sec:numerical-results}
After providing theoretical foundations in Sections \ref{sec:single-source-scenario} and \ref{sec:multi-source-scenario}, in this section we conduct several numerical analyses to corroborate our findings, addressing both the single process-monitor and the multiple process-monitor pairs scenarios.
	
\subsection{Single process-monitor Scenario}

We begin with the single process-monitor problem, for which the analysis in Section \ref{sec:single-source-scenario} proves an ordered policy with two switching thresholds. The purpose of this numerical study is to show how the interaction between the monitor AoI and the base station AoI shapes the optimal three-action policy, and in particular how the state space is divided into sensing, communication, and joint sensing and communication regions.

We solve a truncated version of the MDP on the triangular state space $\alpha^m,\alpha^b\in\{1,2,\dots,A\}$, $\alpha^b\leq\alpha^m$. The optimal value function is computed via value iteration. In the experiment reported below, we set $A=50$, $\gamma=0.85$, $\lambda_0=0.75$, $\lambda_1=0.95$, $\lambda_2=0.65$, $c_0=5$, $c_1=5.5$, and $c_2=6$.
\begin{figure}[ht]
\includegraphics[width=\linewidth]{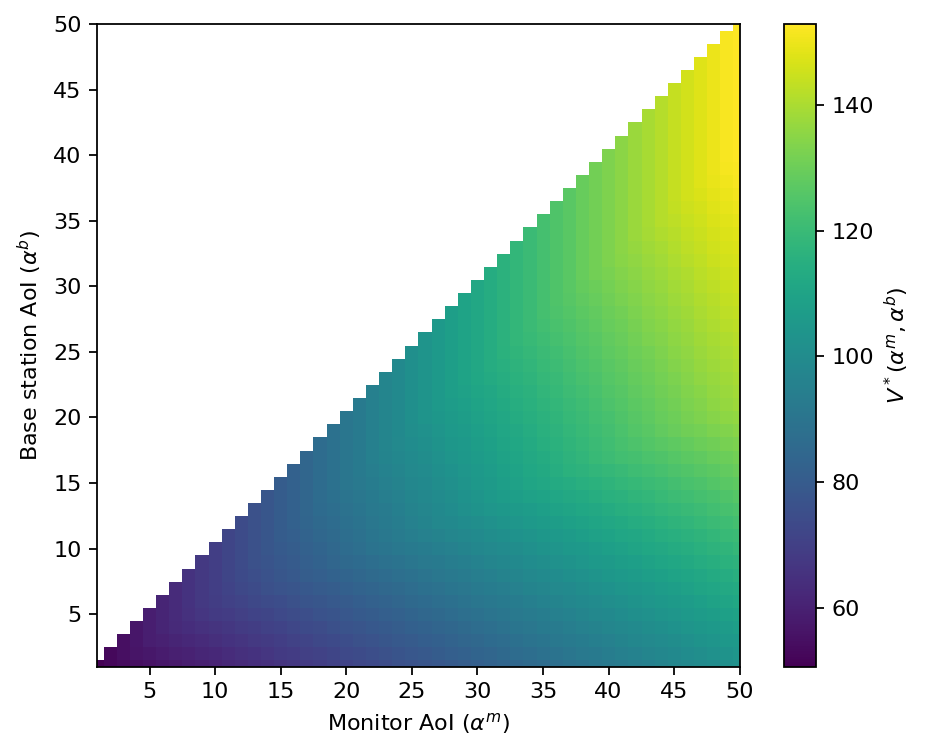}
\centering
\caption{Value function as a function of the monitor and base station AoIs. The value function is non-decreasing in both coordinates and exhibits a structured surface induced by the optimal policy.}
\label{fig:v-heatmap}
\end{figure}
\begin{figure}[ht]
\includegraphics[width=\linewidth]{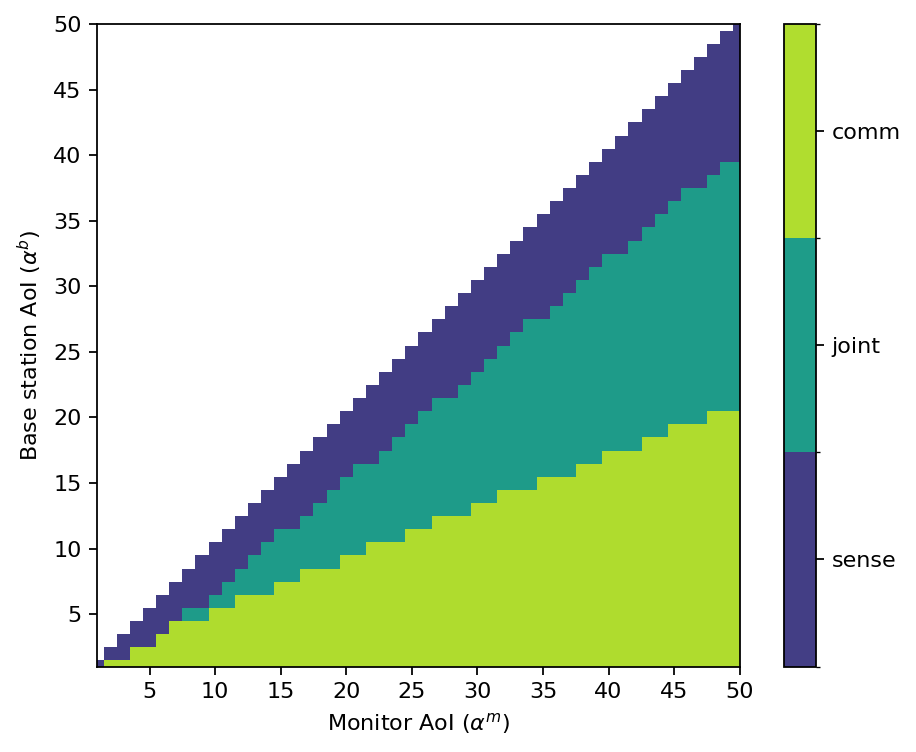}
\centering
\caption{Optimal ISAC action map as a function of the monitor and base station AoIs. The boundaries separating the optimal actions exhibit a threshold structure, consistent with the theoretical results.}
\label{fig:policy-heatmap}
\end{figure}

Figure \ref{fig:v-heatmap} represents the optimal value function $V^*$ over the truncated state space. The value function increases monotonically as either AoI component grows, which is consistent with the fact that stale information leads to a larger long-term cost. Moreover, the dependence on $\alpha^m$ is stronger than that on $\alpha^b$, since the stage cost penalizes the monitor AoI directly, whereas the base station AoI affects performance more indirectly through the quality of the information available for future transmissions.

Figure \ref{fig:policy-heatmap} shows the corresponding optimal action map. Near the diagonal, sensing is optimal, since when the monitor and the base station hold information of comparable age it is preferable to refresh the local estimate before allocating resources to transmission. By contrast, when the monitor AoI becomes sufficiently large, communication is preferred because reducing the age at the monitor becomes the priority. Between these two regimes, there exists an intermediate region in which the joint sensing and communication action is optimal, reflecting the fact that neither pure sensing nor pure communication alone provides the best compromise.

For every fixed value of $\alpha^b$, the optimal action evolves according to the ordered pattern $0\to 2\to 1$, and the lower switching boundary follows the monotonic trend established in Section \ref{sec:single-source-scenario}.

\subsection{Multiple source-monitor pairs Scenario}
For the multiple source-monitor scenario, we evaluate the performance of our approximate Whittle index policy (AWIP) under two different regimes. In the first regime, the sufficient condition for indexability is satisfied, so that the Whittle index policy (WIP) is well defined and can be used as a benchmark. In the second regime, the sufficient condition is violated, so indexability is not guaranteed. In both regimes, we also simulate:
\begin{itemize}
    \item A random policy, which selects $M$ subsystems uniformly at random;
    \item A greedy policy, which selects the $M$ subsystems with the highest $\alpha^m$.
\end{itemize}
All policies apply the local optimal active action once a subsystem is selected, i.e. the action $u\in\{0,1,2\}$ that minimizes $Q_u$ in that subsystem's state.

We consider a truncated state space $\alpha^m,\alpha^b\in\{1,2,\dots,A\}$, $\alpha^b\leq\alpha^m$, with $A=50$ for every subsystem. The local optimal value functions are computed via value iteration. We fix two classes of parameters. The first class has higher action costs and success probabilities ($\lambda_0=0.95$, $\lambda_1=0.98$, $\lambda_2=0.90$, $c_0=7$, $c_1=7$, and $c_2=7$); the second class has lower action costs and success probabilities ($\lambda_0=0.60$, $\lambda_1=0.80$, $\lambda_2=0.55$, $c_0=5$, $c_1=5$, and $c_2=5$).

In our first analysis, we carry out a scalability analysis over $N$. The $N$ sources are split equally in the two classes of parameters. The performance metric is the average discounted cost per source:
\[
J(N):=\frac{1-\gamma}{N}\,\sum_{k=0}^\infty\gamma^k\sum_{i=1}^Ng(S^i_k,u^i_k).
\]
In the online simulations, the infinite-horizon sum is approximated by truncating the evolution at $K_{\max}=200$, and each experiment is repeated $10000$ times.

We first consider a parameter configuration satisfying the sufficient condition for indexability, with $\gamma=0.5$. In this regime, we compare WIP, AWIP, the random policy, and the greedy policy. We fix $M=N/2$, i.e. the base station can address at most half of the total sources. The objective is to verify that AWIP closely tracks the performance of WIP when the Whittle index is well defined, and to assess the gain provided by both index-based policies with respect to the baselines.

\begin{figure}[ht]
\centering
\includegraphics[width=\linewidth]{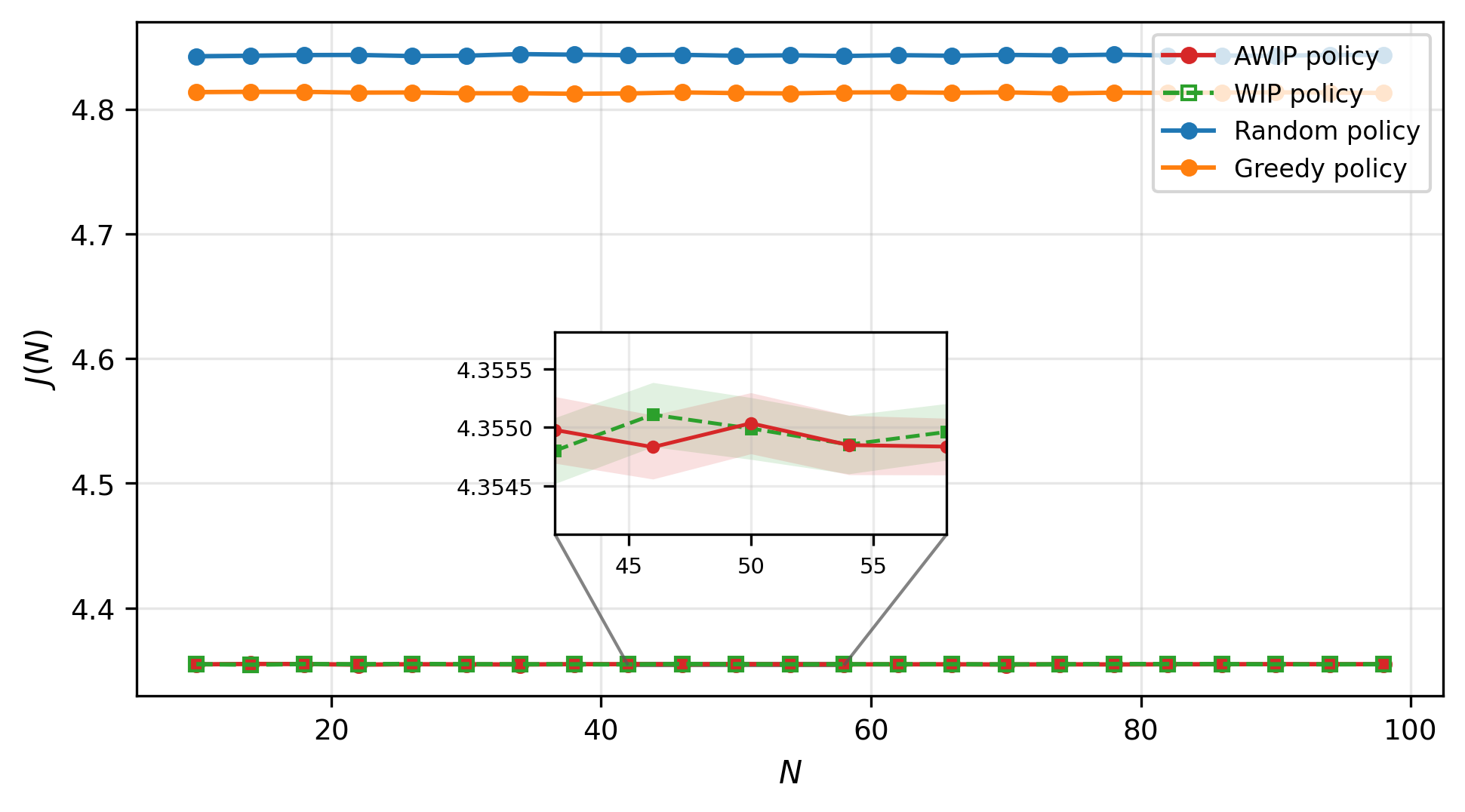}
\caption{Average discounted cost per source $J$ with varying $N$ in the indexable regime: comparison between WIP, AWIP, the random policy, and the greedy policy. $99\%$ confidence intervals are also displayed.}
\label{fig:indexable-multisource}
\end{figure}

Figure \ref{fig:indexable-multisource} shows that, when the sufficient condition for indexability is satisfied, AWIP achieves a performance practically identical to that of WIP, while both clearly outperform the baselines. This supports the use of AWIP as a low-complexity surrogate of WIP in regimes where the exact Whittle index is well defined. Offline, the computation times of the indices for the WIP and the AWIP, shown in Table \ref{tab:offline-wip-awip-timing}, confirm the convenience of using the interpolated indices as the truncation value $A$ grows.

\begin{table}[ht]
\centering
\caption{Offline computation time for WIP and AWIP indices.}
\label{tab:offline-wip-awip-timing}
\begin{tabular}{rrrr}
\toprule
$A$ & WIP (s) & AWIP (s) & Time saving (\%) \\
\midrule
10 & 2.773 & 1.040 & 62.5 \\
20 & 14.32 & 3.421 & 76.1 \\
30 & 39.22 & 7.637 & 80.5 \\
40 & 95.13 & 13.97 & 85.3 \\
50 & 179.7 & 23.81 & 86.7 \\
\bottomrule
\end{tabular}
\end{table}

To complement the large-scale simulations, we also compare the proposed policies with the globally optimal policy on a smaller instance where the full truncated dynamic program associated with \eqref{eq:multi-source-problem} can still be solved exactly. We consider $N=4$ and $M=2$, and repeat the analysis for $A\in\{10,12,15\}$. We use the same two classes of parameters as above. The global optimal value functions are computed via value iteration. Table \ref{tab:exact-DP} reports the centralized optimal cost $J_{\rm DP}$ and the relative gaps of WIP, AWIP, greedy, and random policies.

\begin{table}[ht]
\centering
\caption{Small-scale comparison with the globally optimal centralized policy.}
\label{tab:exact-DP}
\scriptsize
\setlength{\tabcolsep}{3.5pt}
\begin{tabular}{crrrrr}
\toprule
& & \multicolumn{4}{c}{Gap vs. $J_{\rm DP}$ (\%)} \\
\cmidrule(lr){3-6}
$A$ & $J_{\rm DP}$ & WIP & AWIP & Greedy & Random \\
\midrule
10 & 4.354528 & 0.00 & 0.00 & 9.37 & 11.22 \\
12 & 4.354903 & 0.00 & 0.00 & 9.36 & 11.21 \\
15 & 4.354907 & 0.00 & 0.00 & 9.36 & 11.21 \\

\bottomrule
\end{tabular}
\vspace{1mm}
\begin{flushleft}
\end{flushleft}
\end{table}

Table \ref{tab:exact-DP} shows that, in the indexable regime, both WIP and AWIP coincide with the globally optimal policy up to the displayed numerical tolerance. This confirms that the proposed approximate index captures almost all of the scheduling gain of the centralized optimal policy, while retaining the simple online structure used in the large scale simulations.

We then consider a second parameter configuration in which the sufficient condition for indexability is not satisfied, with $\gamma=0.9$. In this case, we compare AWIP with the random policy and the greedy policies. We fix $M=N/5$, i.e. the base station can address at most one fifth of the total sources. The purpose of this experiment is to show that AWIP remains a competitive heuristic beyond the regime covered by the sufficient condition for indexability.

\begin{figure}[ht]
\centering
\includegraphics[width=\linewidth]{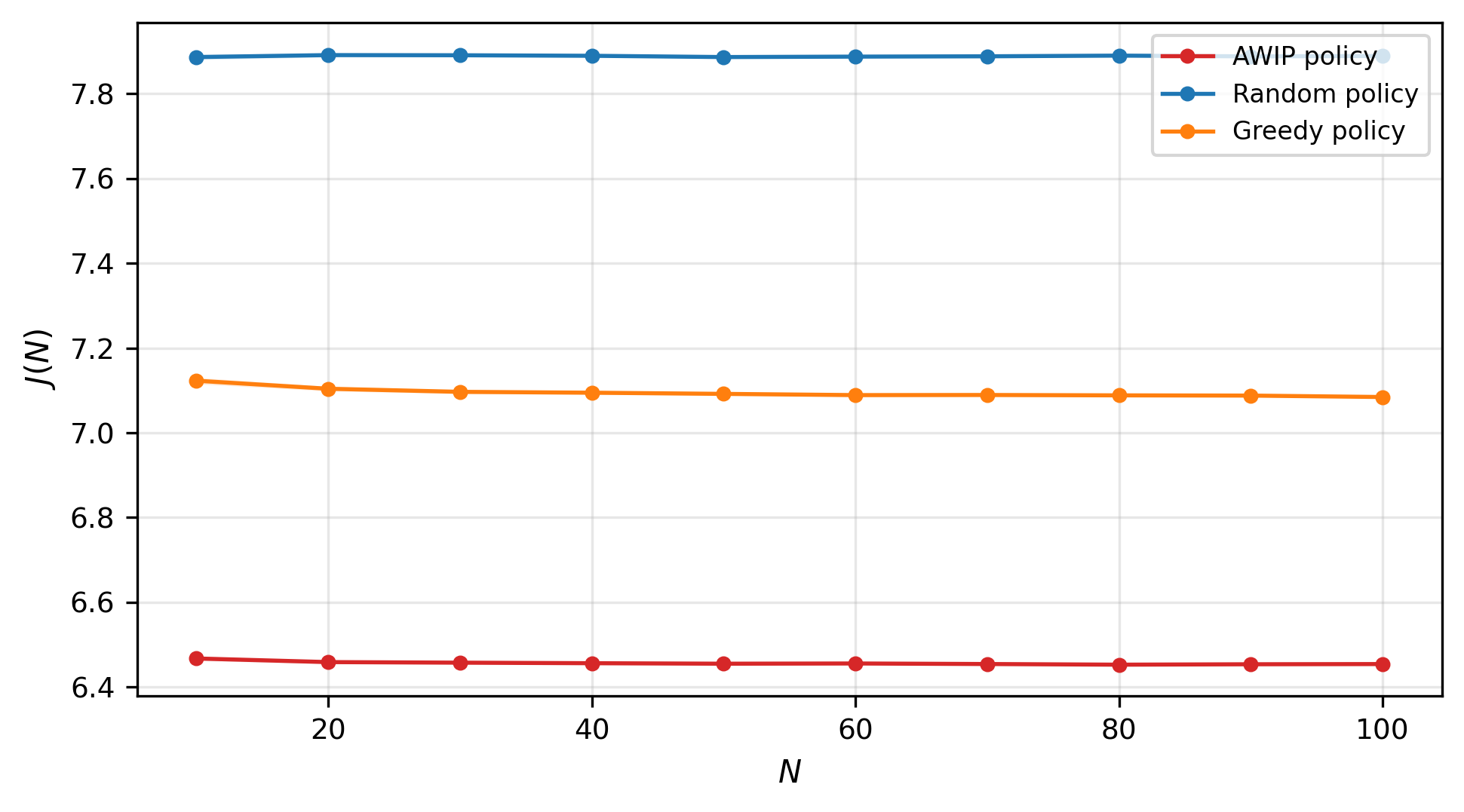}
\caption{Average discounted cost per source $J$ with varying $N$ when the sufficient condition for indexability is violated: comparison between AWIP, the random policy, and the greedy policy. $99\%$ confidence intervals are also displayed.}
\label{fig:nonindexable-multisource}
\end{figure}

Figure \ref{fig:nonindexable-multisource} illustrates that AWIP continues to outperform the baselines when the sufficient condition for indexability is violated. This indicates that the proposed approximation is not only computationally attractive, but also effective in parameter regimes where the Whittle index might not be defined.

\section{Conclusions}\label{sec:conclusions}
In this paper, we studied sensing-communication scheduling in an ISAC architecture for status updating under AoI-based freshness objectives. We formulated the single process-monitor problem as a discounted infinite-horizon Markov decision process, and established that the optimal stationary policy admits a two-threshold structure in the AoI state space, with a monotone lower threshold. Since the AoI state space unbounded, we  quantified the error induced by the truncation of the state space, providing a simple criterion for selecting the truncation level in numerical implementations. For the multiple process-monitor pairs scenario, we formulated the problem as a restless multi-armed bandit and developed scheduling policies based on the Whittle index, including a low-complexity policy with strong performance and a provable approximation bound. Lastly, we carried out numerical analyses to confirm our theoretical findings, showing the threshold geometry of the optimal single-source policy and the effectiveness of the proposed index policies in the multi-source setting.

\bibliographystyle{IEEEtran}
\bibliography{bibliography}

\newpage
\appendices

\section{Proof of Lemma \ref{lemma:Vstar-horizontal-lower-bound}}
\label{app:Vstar-horizontal-lower-bound}

Let $V^{(0)}\equiv 0$ and define $V^{(n+1)}=TV^{(n)}$ for $n\geq 0$. For each $n\geq 0$, set
\[
L_n:=\inf_{\alpha^m\geq \alpha^b\geq 1}\bigl[V^{(n)}(\alpha^m+1,\alpha^b)-V^{(n)}(\alpha^m,\alpha^b)\bigr].
\]
Clearly, $L_0=0$. Since $V^{(n)}$ is coordinatewise nondecreasing, $L_n\geq 0$ for every $n\geq 0$.

Fix $\alpha^m\ge \alpha^b\ge 1$. From the state evolution in
\eqref{eq:state-evolution},
\begin{align*}
Q_0^{(n)}&(\alpha^m+1,\alpha^b)-Q_0^{(n)}(\alpha^m,\alpha^b)\\
&=1+\gamma\lambda_0\Bigl[V^{(n)}(\alpha^m+2,1)-V^{(n)}(\alpha^m+1,1)\Bigr]\\
&\quad+\gamma(1-\lambda_0)\Bigl[V^{(n)}(\alpha^m+2,\alpha^b+1)\\
&\quad -V^{(n)}(\alpha^m+1,\alpha^b+1)\Bigr],\\
Q_1^{(n)}&(\alpha^m+1,\alpha^b)-Q_1^{(n)}(\alpha^m,\alpha^b)\\
&=1+\gamma(1-\lambda_1)\Bigl[V^{(n)}(\alpha^m+2,\alpha^b+1)\\
&\quad -V^{(n)}(\alpha^m+1,\alpha^b+1)\Bigr],\\
Q_2^{(n)}&(\alpha^m+1,\alpha^b)-Q_2^{(n)}(\alpha^m,\alpha^b)\\
&=1+\gamma(1-\lambda_2)\Bigl[V^{(n)}(\alpha^m+2,\alpha^b+1)\\
&\quad -V^{(n)}(\alpha^m+1,\alpha^b+1)\Bigr].
\end{align*}
By definition of $L_n$, every horizontal increment of $V^{(n)}$ is at least $L_n$. Since $\lambda_2\leq\lambda_0\leq\lambda_1$ and $L_n\geq 0$, each of the three differences above is bounded below by $1+\gamma(1-\lambda_1)L_n$.
Using $\min_i x_i-\min_i y_i\geq \min_i (x_i-y_i)$, we get $L_{n+1}\geq 1+\gamma(1-\lambda_1)L_n$. Hence,
\[
L_n\geq \sum_{k=0}^{n-1}\big(\gamma(1-\lambda_1)\big)^k
=\frac{1-\big(\gamma(1-\lambda_1)\big)^n}{1-\gamma+\gamma\lambda_1}.
\]
Pointwise convergence of $V^{(n)}$ to $V^*$ ensures
\[
V^*(\alpha^m+1,\alpha^b)-V^*(\alpha^m,\alpha^b)\geq \frac{1}{1-\gamma+\gamma\lambda_1},
\]
for every $\alpha^m\geq \alpha^b\geq 1$.\qed

\section{Proof of Lemma \ref{lemma:Vstar-local-upper-bound}}
\label{app:Vstar-local-upper-bound}

Fix $(\alpha^m,\alpha^b)\in\mathcal S$ and $(\tilde\alpha^m,\tilde\alpha^b)\in\mathcal S$ such that
\[
0\leq \tilde\alpha^m-\alpha^m\leq 1,\qquad 0\leq \tilde\alpha^b-\alpha^b\leq 1.
\]
Let $V^{(0)}\equiv 0$ and define $V^{(n+1)}=TV^{(n)}$ for $n\geq 0$. For each $n\geq 0$, set
\[
U_n:=\sup\left\{
\begin{aligned}
&V^{(n)}(\tilde\alpha^m,\tilde\alpha^b)-V^{(n)}(\alpha^m,\alpha^b):\\
&\alpha^m\geq \alpha^b\geq 1,\;
0\leq \tilde\alpha^m-\alpha^m\leq 1,\\
&0\leq \tilde\alpha^b-\alpha^b\leq 1
\end{aligned}
\right\}.
\]
Clearly, $U_0=0$. For every $u\in\{0,1,2\}$, the stage-cost difference satisfies
\[
g\bigl((\tilde\alpha^m,\tilde\alpha^b),u\bigr)-g\bigl((\alpha^m,\alpha^b),u\bigr)\leq 1.
\]
Moreover, under the same action $u$ and the same realization of $\eta$, the corresponding next states also differ by at most one in each coordinate. Therefore,
\[
Q_u^{(n)}(\tilde\alpha^m,\tilde\alpha^b)-Q_u^{(n)}(\alpha^m,\alpha^b)\leq 1+\gamma U_n.
\]
Using $V^{(n+1)}(s)=\min_{u\in\{0,1,2\}} Q_u^{(n)}(s)$ and $\min_i x_i-\min_i y_i\leq \max_i (x_i-y_i)$, we get $U_{n+1}\leq 1+\gamma U_n$. By iteration,
\[
U_n\leq \sum_{k=0}^{n-1}\gamma^k=\frac{1-\gamma^n}{1-\gamma}.
\]
Pointwise convergence of $V^{(n)}$ to $V^*$ ensures
\[
V^*(\tilde\alpha^m,\tilde\alpha^b)-V^*(\alpha^m,\alpha^b)\leq \frac{1}{1-\gamma}.
\]
Nonnegativity follows immediately from $V^*$ being coordinatewise nondecreasing.\qed

\section{Proof of Lemma \ref{lemma:ordered-cost-diagonal-zero}}
\label{app:ordered-cost-diagonal-zero}

At $(\alpha,\alpha)$, the Bellman equation gives
\begin{align*}
Q_0^*(\alpha,\alpha)
&=\alpha+c_0+\gamma\lambda_0V^*(\alpha+1,1)\\
&\quad+\gamma(1-\lambda_0)V^*(\alpha+1,\alpha+1),\\
Q_1^*(\alpha,\alpha)
&=\alpha+c_1+\gamma V^*(\alpha+1,\alpha+1),\\
Q_2^*(\alpha,\alpha)
&=\alpha+c_2+\gamma\lambda_2V^*(\alpha+1,1)\\
&\quad+\gamma(1-\lambda_2)V^*(\alpha+1,\alpha+1).
\end{align*}
Hence
\[
\begin{aligned}
Q_0^*(\alpha,\alpha)-Q_1^*(\alpha,\alpha)
&=(c_0-c_1)-\gamma\lambda_0\bigl(V^*(\alpha+1,\alpha+1)\\
&\qquad -V^*(\alpha+1,1)\bigr)\leq 0
\end{aligned}
\]
and
\[
\begin{aligned}
Q_0^*(\alpha,\alpha)-Q_2^*(\alpha,\alpha)
&=(c_0-c_2)-\gamma(\lambda_0-\lambda_2)\\
&\quad\cdot\bigl(V^*(\alpha+1,\alpha+1)\\
&\qquad -V^*(\alpha+1,1)\bigr)\leq 0,
\end{aligned}
\]
where we used $c_0\leq c_1\leq c_2$, $\lambda_2\leq \lambda_0$, and Lemma \ref{lemma:Vstar-in-F}. Therefore action $0$ is optimal on the diagonal.\qed

\section{Proof of Lemma \ref{lemma:ordered-cost-diagonal-gap}}
\label{app:ordered-cost-diagonal-gap}

Iterating \eqref{eq:ordered-cost-diagonal-recursion} $n$ times yields
\[
\begin{aligned}
B(\alpha^b)&=\sum_{j=0}^{n-1} \gamma^j(1-\lambda_0)^j
\bigl[1+\gamma\lambda_0A(\alpha^b+j)\bigr]\\
&\quad+\gamma^n(1-\lambda_0)^n B(\alpha^b+n).
\end{aligned}
\]
By Lemma \ref{lemma:Vstar-local-upper-bound}, $0\leq B(\alpha^b)\leq \frac{1}{1-\gamma}$ for every $\alpha^b$, so the sequence $(B(\alpha^b))_{\alpha^b\geq 1}$ is bounded. Since $\gamma(1-\lambda_0)<1$, letting $n\to\infty$ gives
\[
B(\alpha^b)=\sum_{j=0}^{\infty} \gamma^j(1-\lambda_0)^j\bigl[1+\gamma\lambda_0A(\alpha^b+j)\bigr].
\]
By Lemma \ref{lemma:Vstar-in-F}, the sequence $(A(\alpha^b))_{\alpha^b\geq 1}$ is nonincreasing. Hence $A(\alpha^b+j)\le A(\alpha^b)$ for every $j\in\mathbb N_0$, and therefore
\[
B(\alpha^b)\leq \sum_{j=0}^{\infty} \gamma^j(1-\lambda_0)^j\bigl[1+\gamma\lambda_0A(\alpha^b)\bigr]=\frac{1+\gamma\lambda_0A(\alpha^b)}{1-\gamma+\gamma\lambda_0}.
\]
Subtracting $A(\alpha^b)$ from both sides, we obtain
\[
B(\alpha^b)-A(\alpha^b)\leq \frac{1-(1-\gamma)A(\alpha^b)}{1-\gamma+\gamma\lambda_0}.
\]
Since $A(\alpha^b)\geq 0$, it follows that
\[
B(\alpha^b)-A(\alpha^b) \leq \frac{1}{1-\gamma+\gamma\lambda_0} < \frac{1}{\gamma\lambda_0}.
\]\qed

\section{Proof of Lemma \ref{lemma:ordered-cost-vertical-bound}}
\label{app:ordered-cost-vertical-bound}
Let
\[
M(\alpha^b):=\sup_{\alpha^m\ge \alpha^b+1} C(\alpha^m,\alpha^b).
\]
Fix $\alpha^b\ge 1$ and $\alpha^m\ge \alpha^b+1$. Using the Bellman equation and $\min_i x_i-\min_i y_i\le \max_i(x_i-y_i)$, we obtain
\begin{align*}
C(\alpha^m,\alpha^b)
&\le \max\Bigl\{Q_0^*(\alpha^m+1,\alpha^b+2)\\
&\quad -Q_0^*(\alpha^m+1,\alpha^b+1),\\
&\quad Q_1^*(\alpha^m+1,\alpha^b+2)\\
&\quad -Q_1^*(\alpha^m+1,\alpha^b+1),\\
&\quad Q_2^*(\alpha^m+1,\alpha^b+2)\\
&\quad -Q_2^*(\alpha^m+1,\alpha^b+1)\Bigr\}.
\end{align*}

For $u=0$,
\begin{align*}
&Q_0^*(\alpha^m+1,\alpha^b+2)-Q_0^*(\alpha^m+1,\alpha^b+1)\\
&\qquad=\gamma(1-\lambda_0)\Bigl[V^*(\alpha^m+2,\alpha^b+3)\\
&\qquad\quad -V^*(\alpha^m+2,\alpha^b+2)\Bigr]\\
&\qquad=\gamma(1-\lambda_0)\,C(\alpha^m+1,\alpha^b+1).
\end{align*}

For $u=1$,
\begin{align*}
&Q_1^*(\alpha^m+1,\alpha^b+2)-Q_1^*(\alpha^m+1,\alpha^b+1)\\
&\quad=\gamma\lambda_1\Bigl[V^*(\alpha^b+3,\alpha^b+3)\\
&\qquad -V^*(\alpha^b+2,\alpha^b+2)\Bigr]\\
&\qquad+\gamma(1-\lambda_1)\Bigl[V^*(\alpha^m+2,\alpha^b+3)\\
&\qquad -V^*(\alpha^m+2,\alpha^b+2)\Bigr]\\
&\quad =\gamma\lambda_1 B(\alpha^b+2)\\
&\qquad+\gamma(1-\lambda_1)C(\alpha^m+1,\alpha^b+1).
\end{align*}

For $u=2$,
\begin{align*}
&Q_2^*(\alpha^m+1,\alpha^b+2)-Q_2^*(\alpha^m+1,\alpha^b+1)\\
&\quad=\gamma\lambda_2\Bigl[V^*(\alpha^b+3,1)-V^*(\alpha^b+2,1)\Bigr]\\
&\qquad+\gamma(1-\lambda_2)\Bigl[V^*(\alpha^m+2,\alpha^b+3)\\
&\qquad -V^*(\alpha^m+2,\alpha^b+2)\Bigr]\\
&\quad=\gamma\lambda_2 A(\alpha^b+1)\\
&\qquad+\gamma(1-\lambda_2)C(\alpha^m+1,\alpha^b+1).
\end{align*}
Taking the supremum over $\alpha^m\ge \alpha^b+1$, we obtain
\begin{align*}
M(\alpha^b)\le \max\Bigl\{
&\gamma(1-\lambda_0)M(\alpha^b+1),\\
&\gamma\lambda_1 B(\alpha^b+2)+\gamma(1-\lambda_1)M(\alpha^b+1),\\
&\gamma\lambda_2 A(\alpha^b+1)+\gamma(1-\lambda_2)M(\alpha^b+1)\Bigr\}.
\end{align*}
We compare each term in the right-hand side with $B(\alpha^b+1)$. First,
\begin{align*}
&\gamma(1-\lambda_0)M(\alpha^b+1)-B(\alpha^b+1)\\
&\qquad=\gamma(1-\lambda_0)\bigl(M(\alpha^b+1)-B(\alpha^b+2)\bigr)\\
&\qquad\quad+\bigl[\gamma(1-\lambda_0)B(\alpha^b+2)-B(\alpha^b+1)\bigr].
\end{align*}
From \eqref{eq:ordered-cost-diagonal-recursion}, the bracket is nonpositive. Hence
\[
\begin{aligned}
&\gamma(1-\lambda_0)M(\alpha^b+1)-B(\alpha^b+1)\\
&\qquad\leq \gamma(1-\lambda_0)\bigl(M(\alpha^b+1)-B(\alpha^b+2)\bigr).
\end{aligned}
\]
Second,
\begin{align*}
&\gamma\lambda_1B(\alpha^b+2)+\gamma(1-\lambda_1)M(\alpha^b+1)-B(\alpha^b+1)\\
&\qquad=\gamma(1-\lambda_1) \bigl(M(\alpha^b+1)-B(\alpha^b+2)\bigr)\\
&\qquad\quad +\bigl[\gamma B(\alpha^b+2)-B(\alpha^b+1)\bigr].
\end{align*}
Moreover,
\[
\gamma B(\alpha^b+2)-B(\alpha^b+1)
=
-1+\gamma\lambda_0\bigl(B(\alpha^b+2)-A(\alpha^b+1)\bigr).
\]
Since $(A(\alpha^b))_{\alpha^b\ge 1}$ is nonincreasing, Lemma 7 gives
\[
B(\alpha^b+2)-A(\alpha^b+1)\le B(\alpha^b+2)-A(\alpha^b+2)<\frac{1}{\gamma\lambda_0},
\]
so the bracket is strictly negative. Therefore
\[
\begin{aligned}
&\gamma\lambda_1B(\alpha^b+2)+\gamma(1-\lambda_1)M(\alpha^b+1)-B(\alpha^b+1)\\
&\qquad\le
\gamma(1-\lambda_1)\bigl(M(\alpha^b+1)-B(\alpha^b+2)\bigr).
\end{aligned}
\]
Third,
\begin{align*}
&\gamma\lambda_2A(\alpha^b+1)+\gamma(1-\lambda_2)M(\alpha^b+1)-B(\alpha^b+1)\\
&\qquad=\gamma(1-\lambda_2)
\bigl(M(\alpha^b+1)-B(\alpha^b+2)\bigr)\\
&\qquad\quad+\bigl[\gamma\lambda_2A(\alpha^b+1)
+\gamma(1-\lambda_2)B(\alpha^b+2)\\
&\qquad\quad -B(\alpha^b+1)\bigr].
\end{align*}
Again using \eqref{eq:ordered-cost-diagonal-recursion},
\[
\begin{aligned}
&\gamma\lambda_2A(\alpha^b+1)+\gamma(1-\lambda_2)B(\alpha^b+2)-B(\alpha^b+1)\\
&\qquad=-1+\gamma(\lambda_0-\lambda_2)
\bigl(B(\alpha^b+2)-A(\alpha^b+1)\bigr).
\end{aligned}
\]
By the same bound as above,
\[
\gamma(\lambda_0-\lambda_2)\bigl(B(\alpha^b+2)-A(\alpha^b+1)\bigr)
<
\frac{\lambda_0-\lambda_2}{\lambda_0}\le 1,
\]
hence this bracket is also strictly negative. Therefore
\[
\begin{aligned}
&\gamma\lambda_2A(\alpha^b+1)+\gamma(1-\lambda_2)M(\alpha^b+1)-B(\alpha^b+1)\\
&\qquad\le
\gamma(1-\lambda_2)\bigl(M(\alpha^b+1)-B(\alpha^b+2)\bigr).
\end{aligned}
\]

Combining the three bounds, we get
\begin{align*}
M(\alpha^b)-B(\alpha^b+1)
&\leq \max\Bigl\{\gamma(1-\lambda_0)\bigl(M(\alpha^b+1)\\
&\qquad -B(\alpha^b+2)\bigr),\\
&\quad \gamma(1-\lambda_1)\bigl(M(\alpha^b+1)\\
&\qquad -B(\alpha^b+2)\bigr),\\
&\quad \gamma(1-\lambda_2)\bigl(M(\alpha^b+1)\\
&\qquad -B(\alpha^b+2)\bigr) \Bigr\}.
\end{align*}
Taking positive parts, let $Y(\alpha^b) := \Bigl[M(\alpha^b)-B(\alpha^b+1)\Bigr]^+$. Then
\begin{align*}
Y(\alpha^b) &\leq \Biggl[\max\Bigl\{\gamma(1-\lambda_0)\bigl(M(\alpha^b+1)-B(\alpha^b+2)\bigr),\\
&\qquad\qquad \gamma(1-\lambda_1)\bigl(M(\alpha^b+1)-B(\alpha^b+2)\bigr),\\
&\qquad\qquad \gamma(1-\lambda_2)\bigl(M(\alpha^b+1)-B(\alpha^b+2)\bigr)\Bigr\}\Biggr]^+\\
&\leq \max\Bigl\{ \gamma(1-\lambda_0),\gamma(1-\lambda_1),\gamma(1-\lambda_2)\Bigr\}\\
&\qquad\cdot\Bigl[M(\alpha^b+1)-B(\alpha^b+2)\Bigr]^+.
\end{align*}
Since $\lambda_2\le \lambda_0\le \lambda_1$, this gives
\[
Y(\alpha^b)\le \gamma(1-\lambda_2)Y(\alpha^b+1),\qquad \forall \alpha^b\ge 1.
\]

By Lemma 5, both $C(\alpha^m,\alpha^b)$ and $B(\alpha^b+1)$ are bounded between $0$ and
$1/(1-\gamma)$, so $(Y(\alpha^b))_{\alpha^b\ge 1}$ is bounded. If $Y(\bar\alpha^b)>0$ for some
$\bar\alpha^b$, then iterating the previous inequality gives
\[
Y(\bar\alpha^b+n)\ge \frac{Y(\bar\alpha^b)}{[\gamma(1-\lambda_2)]^n},
\qquad \forall n\ge 1,
\]
which is impossible because $\gamma(1-\lambda_2)<1$ and $(Y(\alpha^b))_{\alpha^b\ge 1}$ is bounded.
Therefore $Y(\alpha^b)=0$ for every $\alpha^b\ge 1$, that is,
\[
M(\alpha^b)\le B(\alpha^b+1),\qquad \forall \alpha^b\ge 1.
\]
Since $C(\alpha^m,\alpha^b)\le M(\alpha^b)$ by definition of $M(\alpha^b)$, we conclude
\[
C(\alpha^m,\alpha^b)\le B(\alpha^b+1),
\]
for every $\alpha^b\ge 1$ and every $\alpha^m\ge \alpha^b+1$.\qed

\section{Proof of Lemma \ref{lemma:delta-non-decreasing-in-a_s}}
\label{app:delta-non-decreasing-in-a_s}

Fix $\alpha^b$. From \eqref{eq:deltas},
\begin{align*}
&\Delta^*_{01}(\alpha^m+1,\alpha^b)-\Delta^*_{01}(\alpha^m,\alpha^b)\\
&\quad=\gamma\lambda_0\bigl[V^*(\alpha^m+2,1)-V^*(\alpha^m+1,1)\bigr]\\
&\qquad+\gamma(\lambda_1-\lambda_0)\Bigl[
V^*(\alpha^m+2,\alpha^b+1)\\
&\qquad -V^*(\alpha^m+1,\alpha^b+1)\Bigr],
\end{align*}
which is nonnegative because $V^*$ is coordinatewise nondecreasing and $\lambda_1\geq \lambda_0\geq 0$.

For $\Delta^*_{02}$, again from \eqref{eq:deltas},
\begin{align*}
&\Delta^*_{02}(\alpha^m+1,\alpha^b)-\Delta^*_{02}(\alpha^m,\alpha^b)\\
&\quad=\gamma\lambda_0\bigl[V^*(\alpha^m+2,1)-V^*(\alpha^m+1,1)\bigr]\\
&\qquad+\gamma(\lambda_2-\lambda_0)\Bigl[
V^*(\alpha^m+2,\alpha^b+1)\\
&\qquad -V^*(\alpha^m+1,\alpha^b+1)\Bigr].
\end{align*}
By Lemma \ref{lemma:Vstar-horizontal-lower-bound},
\[
V^*(\alpha^m+2,1)-V^*(\alpha^m+1,1) \geq \frac{1}{1-\gamma+\gamma\lambda_1},
\]
and by Lemma \ref{lemma:Vstar-local-upper-bound},
\[
0\leq V^*(\alpha^m+2,\alpha^b+1)-V^*(\alpha^m+1,\alpha^b+1)\leq \frac{1}{1-\gamma}.
\]
Combining these bounds, $\Delta^*_{02}(\alpha^m+1,\alpha^b)-\Delta^*_{02}(\alpha^m,\alpha^b)$ is nonnegative if
\[
\frac{\lambda_0}{1-\gamma+\gamma\lambda_1} \geq \frac{\lambda_0-\lambda_2}{1-\gamma},
\]
which is equivalent to \eqref{eq:parameters-assumption}.

Finally,
\[
\begin{aligned}
&\Delta^*_{21}(\alpha^m+1,\alpha^b)-\Delta^*_{21}(\alpha^m,\alpha^b)\\
&\quad=\gamma(\lambda_1-\lambda_2)\Bigl[
V^*(\alpha^m+2,\alpha^b+1)\\
&\qquad -V^*(\alpha^m+1,\alpha^b+1)\Bigr]
\geq 0,
\end{aligned}
\]
because $\lambda_1\geq \lambda_2$ and $V^*$ is coordinatewise nondecreasing.\qed

\section{Proof of Lemma \ref{lemma:delta-non-increasing-in-a_b}}
\label{app:delta-non-increasing-in-a_b}

Fix $\alpha^m$. Using \eqref{eq:deltas}, for $\Delta^*_{01}$ we obtain
\begin{align*}
&\Delta^*_{01}(\alpha^m,\alpha^b+1)-\Delta^*_{01}(\alpha^m,\alpha^b)\\
&\quad=
-\gamma\lambda_1\Bigl[
V^*(\alpha^b+2,\alpha^b+2)-V^*(\alpha^b+1,\alpha^b+1)
\Bigr]\\
&\qquad
+\gamma(\lambda_1-\lambda_0)\Bigl[
V^*(\alpha^m+1,\alpha^b+2)\\
&\qquad -V^*(\alpha^m+1,\alpha^b+1)
\Bigr]\\
&\quad =\gamma\Bigl[(\lambda_1-\lambda_0)C(\alpha^m,\alpha^b)-\lambda_1B(\alpha^b+1)\Bigr].
\end{align*}
By Lemma \ref{lemma:ordered-cost-vertical-bound}, $C(\alpha^m,\alpha^b)\le B(\alpha^b+1)$, therefore
\begin{align*}
&\Delta^*_{01}(\alpha^m,\alpha^b+1)-\Delta^*_{01}(\alpha^m,\alpha^b)\\
&\quad\le
\gamma\Bigl[(\lambda_1-\lambda_0)B(\alpha^b+1)-\lambda_1B(\alpha^b+1)\Bigr]\\
&\quad=
-\gamma\lambda_0B(\alpha^b+1)\le 0,
\end{align*}
because $B(\alpha^b+1)\ge 0$ by coordinatewise monotonicity of $V^*$.

For $\Delta^*_{02}$, again using \eqref{eq:deltas}, we obtain
\begin{align*}
&\Delta^*_{02}(\alpha^m,\alpha^b+1)-\Delta^*_{02}(\alpha^m,\alpha^b)\\
&\quad= -\gamma\lambda_2\Bigl[V^*(\alpha^b+2,1)-V^*(\alpha^b+1,1)\Bigr]\\
&\qquad +\gamma(\lambda_2-\lambda_0)\Bigl[V^*(\alpha^m+1,\alpha^b+2)\\
&\qquad -V^*(\alpha^m+1,\alpha^b+1)\Bigr]\\
&\quad= -\gamma\lambda_2A(\alpha^b)+\gamma(\lambda_2-\lambda_0)C(\alpha^m,\alpha^b).
\end{align*}
Since $A(\alpha^b)\geq 0$, $C(\alpha^m,\alpha^b)\geq 0$, and $\lambda_2\leq \lambda_0$, it follows that
\[
\Delta^*_{02}(\alpha^m,\alpha^b+1)-\Delta^*_{02}(\alpha^m,\alpha^b)\leq 0.
\]

Hence both $\Delta^*_{01}(\alpha^m,\alpha^b)$ and $\Delta^*_{02}(\alpha^m,\alpha^b)$ are
nonincreasing in $\alpha^b$ for $1\le \alpha^b\le \alpha^m-1$.\qed

\end{document}